\documentclass[11pt]{article}

\usepackage[T1]{fontenc}
\usepackage[utf8]{inputenc}
\usepackage{lmodern}
\usepackage{microtype}

\usepackage{amsmath,amssymb,amsthm}
\usepackage{mathtools}

\usepackage[letterpaper,margin=1in]{geometry}
\usepackage{graphicx}
\usepackage{booktabs}
\usepackage{array}
\usepackage{tikz}
\usetikzlibrary{arrows,positioning,fit,calc,backgrounds,shapes.geometric}

\usepackage[round,authoryear]{natbib}

\usepackage{xcolor}
\definecolor{linknavy}{HTML}{1F3A93}
\usepackage[colorlinks=true,linkcolor=linknavy,citecolor=linknavy,urlcolor=linknavy]{hyperref}

\theoremstyle{plain}
\newtheorem{theorem}{Theorem}[section]
\newtheorem{proposition}[theorem]{Proposition}

\theoremstyle{definition}
\newtheorem{definition}[theorem]{Definition}
\newtheorem{example}[theorem]{Example}

\theoremstyle{remark}
\newtheorem{remark}[theorem]{Remark}

\theoremstyle{remark}
\newtheorem{outhread}{Ornstein--Uhlenbeck thread}
\newtheorem{cirthread}{Cox--Ingersoll--Ross thread}

\newcommand{\bridge}[1]{\smallskip\noindent\textit{Bridge.}~#1\smallskip}

\newcommand{\E}{\mathbb{E}}
\newcommand{\PP}{\mathbb{P}}
\newcommand{\QQ}{\mathbb{Q}}
\newcommand{\R}{\mathbb{R}}
\newcommand{\Var}{\operatorname{Var}}

\newcommand{\Law}{\operatorname{Law}}
\newcommand{\KL}{\operatorname{KL}}
\newcommand{\Id}{\operatorname{Id}}
\newcommand{\dd}{\,\mathrm{d}}
\newcommand{\e}{\mathrm{e}}
\newcommand{\inner}[2]{\langle #1, #2\rangle}
\newcommand{\norm}[1]{\lVert #1\rVert}
\DeclareMathOperator*{\argmin}{arg\,min}

\newcolumntype{L}[1]{>{\raggedright\arraybackslash}p{#1}}

\title{\textbf{The Mathematics of Modeling the Future}\\[3pt]
\large Filtrations, Forecasts, Markov Semigroups,\\ and the Forward Evolution of Laws}
\author{Miquel Noguer i Alonso\\
\normalsize Artificial Intelligence Finance Institute}
\date{June 18, 2026}

\begin{document}
\maketitle

\begin{abstract}
Modeling the future means specifying conditional laws relative to an evolving information flow and describing how those laws move across time. This paper gives a unified mathematical synthesis of that problem along a single spine. Filtrations encode what is known; conditional expectation gives the optimal point forecast; regular conditional probabilities represent full distributional forecasts; Markov kernels and semigroups propagate observables and laws; and infinitesimal generators encode local dynamics, producing the backward and forward Kolmogorov equations and, in continuous time, stochastic differential equations. Along this spine, martingales isolate irreducible surprise, filtering handles partial observation, finance prices discounted futures under martingale measures, stochastic control chooses among feasible futures, and ergodic theory and large deviations describe the structure of the far future.

The contribution is architectural rather than encyclopedic. The paper makes explicit the connective derivations that turn these classical objects into one forecasting calculus: the tower property becomes the semigroup law; It\^o's formula becomes the backward equation after conditioning; integration by parts gives the forward operator; and generator perturbations become forecast and model-risk distortions. This perspective clarifies that forecasting is not merely extrapolation from data, but the construction of a dynamically coherent family of conditional distributions whose evolution is constrained by information, geometry, and admissible model classes.

Two contrasting closed-form threads are carried throughout: the Gaussian Ornstein--Uhlenbeck process and the non-Gaussian, positivity-preserving Cox--Ingersoll--Ross process. These examples show how the same abstract machinery produces explicit transition laws, invariant distributions, spectral decompositions, likelihoods, term-structure formulae, and long-run asymptotics in different geometries. The synthesis also recasts density evolution as a Wasserstein gradient flow of free energy, places forecasting and estimation in Hilbert, Fisher--Rao, and Wasserstein geometries, gives a discrete-time dictionary for empirical practice, and emphasizes the practical caveat that every forecast inherits uncertainty from the estimated generator. The result is a compact mathematical map from information to prediction, from local dynamics to global laws, and from idealized stochastic models to the model-risk problems encountered in empirical forecasting.
\end{abstract}


\newpage
\tableofcontents
\newpage

\section{Introduction: the future as a conditional law}

A mathematical model of the future must answer two questions. First, relative to what information is the forecast being made? Second, how does the conditional law move forward once that information is fixed? The first question is informational and is encoded by a \emph{filtration}. The second is dynamical and is encoded by a \emph{transition mechanism}: a kernel, a semigroup, a generator, an SDE, or a PDE.

The most economical definition is the following. At time $t$, given an information set $\mathcal{F}_t$, the future of a process $X=(X_s)_{s\ge 0}$ is the conditional law
\[
\Law\!\big((X_{t+u})_{u\ge 0}\mid \mathcal{F}_t\big).
\]
Point forecasts are summaries of this conditional law, most importantly conditional expectations. Distributional forecasts require the whole regular conditional probability. Dynamics specify how these conditional laws change as $t$ and the forecast horizon move. The paper follows the standard progression from information to dynamics:
\[
\text{filtration}\to\text{conditional law}\to\text{Markov kernel}\to\text{semigroup}\to\text{generator}\to\text{PDE/SDE}.
\]
The martingale component then isolates the unpredictable part of the future. Filtering treats the case in which the state is only indirectly observed. Finance changes the probability measure so that no-arbitrage prices become martingales. Long-run theory describes the far future through invariant measures, ergodic averages, and large-deviation rate functions.

\paragraph{What this paper is, and what it adds.} The individual building blocks below are classical. The contribution is a usable \emph{synthesis}, summarized formally in Theorem~\ref{thm:conditional-law-spine}, and it goes past mere collection in five deliberate ways. (i) It states the \emph{connective derivations} that turn the list into a single object: the tower property is the semigroup law, taking expectations in It\^o's formula is the backward Kolmogorov equation, integration by parts is the forward operator, and the filtering and control Riccati equations are dual under the standard linear--quadratic--Gaussian transposition/time-reversal correspondence. (ii) It carries \emph{two contrasting worked examples} through every section --- the Gaussian Ornstein--Uhlenbeck (OU) process and the non-Gaussian, positivity-preserving Cox--Ingersoll--Ross (CIR) process --- so each abstract object acquires a closed form and the construction is visibly not Gaussian-specific. (iii) It recasts the forward evolution of densities as a \emph{Wasserstein gradient flow} of free energy, delivering on ``the forward evolution of laws'' geometrically rather than only analytically (\S\ref{sec:wasserstein}). (iv) It gives a \emph{discrete-time dictionary} for the whole spine, the form most forecasting and machine-learning practice actually uses (\S\ref{sec:discrete}). (v) It closes the loop to data: every object is estimated, and a forecast is only as good as the estimated generator (\S\ref{sec:estimation}). Throughout, results are stated in representative form; the hypotheses that matter are recorded in \S\ref{sec:synthesis}.

\section{Related work and what the synthesis adds}\label{sec:related}

The progression information $\to$ kernel $\to$ semigroup $\to$ generator $\to$ PDE is the organizing logic of the standard references on Markov processes \citep{EthierKurtz1986,RogersWilliams2000,RevuzYor1999}, and the semigroup-first development is the explicit program of \citet{Pazy1983} on the analytic side and \citet{Bobrowski2005} on the probabilistic side. The duality between propagating observables and propagating measures is classical \citep{EthierKurtz1986}; the curvature and functional-inequality view of the generator is developed in \citet{BGL2014}; the metric-geometry view of the forward equation originates with \citet{JKO1998} and \citet{Otto2001} and is systematized in \citet{AGS2008} and \citet{Villani2009}. The financial half rests on \citet{DelbaenSchachermayer2006}, \citet{KaratzasShreve1991}, \citet{Bjork2009}, and \citet{Duffie2001}; the operator/spectral view of long-horizon pricing is \citet{HansenScheinkman2009}.

Against this background the present synthesis adds four things that no single one of those sources foregrounds simultaneously. First, it keeps \emph{one spine} and shows that the transitions between its links are not analogies but identities (\S\ref{sec:semigroups}--\S\ref{sec:kolmogorov}, \S\ref{sec:control}). Second, it makes the spine \emph{concrete} with two closed-form threads, OU and CIR, so that the generator, both Kolmogorov equations, invariant law, spectral structure, filtering, term-structure pricing, control, and long-run behavior are instantiated in parallel. The OU thread gives the Gaussian/Hermite/Vasicek case; the CIR thread gives the positive/Gamma/Laguerre/Riccati case. Third, it places the \emph{analytic} forward equation (Fokker--Planck) and its \emph{geometric} form (Wasserstein gradient flow) side by side. Fourth, it supplies a \emph{discrete-time} column and an explicit \emph{estimation} caveat, which is where forecasting is done in practice. The intended contribution is therefore pedagogical and architectural: a compressed, self-consistent map rather than a new frontier theorem.

\section{Information: filtered probability spaces}\label{sec:info}

Fix a probability space $(\Omega,\mathcal{F},\PP)$ carrying a filtration $(\mathcal{F}_t)_{t\ge 0}$ satisfying the usual conditions: completeness and right-continuity. A stochastic process is a family $X=(X_t)_{t\ge 0}$ of random variables valued in a Polish state space $(E,\mathcal{E})$; throughout, unless stated otherwise, $E=\R^d$ with its Borel $\sigma$-algebra. The filtration is the formal carrier of time:
\[
\mathcal{F}_t=\text{``everything observable up to and including time }t\text{''},\qquad
\mathcal{F}_\infty=\sigma\Big(\textstyle\bigcup_{t\ge 0}\mathcal{F}_t\Big).
\]
The past at $t$ is $\mathcal{F}_t$; the future is the family $\{X_s:s>t\}$, whose law conditioned on $\mathcal{F}_t$ is the object to be modeled.

\begin{definition}[Adapted, progressive, optional, and predictable processes]
A process $X$ is \emph{adapted} to $(\mathcal{F}_t)_{t\ge0}$ if $X_t$ is $\mathcal{F}_t$-measurable for every $t$. It is \emph{progressively measurable} if, for every $T$, the restriction of $(\omega,t)\mapsto X_t(\omega)$ to $\Omega\times[0,T]$ is $\mathcal{F}_T\otimes\mathcal{B}([0,T])$-measurable. It is \emph{optional} if it is measurable with respect to the optional $\sigma$-algebra, and \emph{predictable} if it is measurable with respect to the $\sigma$-algebra generated by left-continuous adapted processes.
\end{definition}

Adaptedness says that present values use no future information. Predictability says that a quantity is known just before it acts. This distinction is central in stochastic integration: integrands are predictable, while the next martingale increment is not. In that sense, predictability separates the systematic part of a model from its innovation. For the measure-theoretic foundations see \citet{Kallenberg2021} and \citet{DellacherieMeyer1982}.

\begin{definition}[Natural and enlarged filtrations]
The natural filtration of $X$ is $\mathcal{F}^X_t=\sigma(X_s:0\le s\le t)$, usually augmented to satisfy the usual conditions. If $\mathcal{F}^X_t\subseteq\mathcal{G}_t$, then $(\mathcal{G}_t)$ is an enlarged filtration. The same process may have different forecasting properties under different filtrations.
\end{definition}

\begin{remark}[Information is part of the model]
A claim such as ``$X$ is unpredictable'' is never absolute. It means unpredictable \emph{relative to a filtration}. Enlarging $\mathcal{F}_t$ may transform a martingale into a process with drift, or a Markov process into a non-Markov process relative to the larger information set. A forecasting model is therefore not merely a law for $X$; it is a pair consisting of a law and an information structure.
\end{remark}

\section{Conditional laws and optimal prediction}\label{sec:conditional}

The full object of prediction is a conditional distribution. The most important scalar summary of that distribution is the conditional expectation.

\begin{definition}[Conditional expectation]
Let $Z\in L^1(\Omega,\mathcal{F},\PP)$ and let $\mathcal{G}\subseteq\mathcal{F}$ be a sub-$\sigma$-algebra. The conditional expectation $\E[Z\mid\mathcal{G}]$ is the almost surely unique $\mathcal{G}$-measurable integrable random variable satisfying
\begin{equation}\label{eq:condexp}
\int_G \E[Z\mid\mathcal{G}]\dd\PP=\int_G Z\dd\PP,\qquad G\in\mathcal{G}.
\end{equation}
\end{definition}

Existence and uniqueness follow from the Radon--Nikodym theorem \citep[Ch.~9]{Williams1991}. For square-integrable random variables, conditional expectation is a geometric projection.

\begin{theorem}[Conditional expectation is the best mean-square predictor]\label{thm:projection}
If $Z\in L^2(\Omega,\mathcal{F},\PP)$, then
\[
\E[Z\mid\mathcal{G}]=\argmin_{Y\in L^2(\Omega,\mathcal{G},\PP)}\E\big[(Z-Y)^2\big].
\]
Equivalently, the residual $Z-\E[Z\mid\mathcal{G}]$ is orthogonal to every $\mathcal{G}$-measurable $Y\in L^2$: $\E[(Z-\E[Z\mid\mathcal{G}])\,Y]=0$.
\end{theorem}

\begin{proof}[Proof sketch]
The space $L^2(\Omega,\mathcal{G},\PP)$ is a closed linear subspace of $L^2(\Omega,\mathcal{F},\PP)$. Orthogonal projection onto this subspace exists and is unique. The defining identity of the projection is precisely the conditional-expectation identity \eqref{eq:condexp}, tested against indicators and extended by density.
\end{proof}

Among all forecasts allowed to use the information $\mathcal{G}$, the conditional expectation is optimal under squared loss. Other scoring rules produce other summaries, but conditional expectation is the canonical Hilbert-space forecast.

\begin{proposition}[Tower property and variance decomposition]\label{prop:tower}
For sub-$\sigma$-algebras $\mathcal{H}\subseteq\mathcal{G}\subseteq\mathcal{F}$ and $Z\in L^1$,
\[
\E\big[\E[Z\mid\mathcal{G}]\,\big|\,\mathcal{H}\big]=\E[Z\mid\mathcal{H}].
\]
If $Z\in L^2$, then $\Var(Z)=\E[\Var(Z\mid\mathcal{G})]+\Var\big(\E[Z\mid\mathcal{G}]\big)$.
\end{proposition}

The tower property says that a coarser forecast equals the conditional average of a finer one; it is the consistency condition that the whole dynamical theory will inherit (\S\ref{sec:semigroups}). The variance decomposition splits total uncertainty into what is resolved by $\mathcal{G}$ and what remains. Operationally, modeling the future means choosing and updating an information set that resolves useful uncertainty without introducing inaccessible information.

\begin{definition}[Regular conditional law]
Let $X$ be a random element of a Polish space $E$. A regular conditional law of $X$ given $\mathcal{G}$ is a probability kernel $K:\Omega\times\mathcal{E}\to[0,1]$ such that $K(\cdot,A)$ is $\mathcal{G}$-measurable for every $A\in\mathcal{E}$ and $K(\omega,A)=\PP(X\in A\mid\mathcal{G})(\omega)$ a.s.
\end{definition}

For Polish state spaces, regular conditional laws exist. They upgrade point forecasting to distributional forecasting: $\E[f(X)\mid\mathcal{G}](\omega)=\int_E f(x)\,K(\omega,\dd x)$ for bounded measurable $f$. The full future forecast is thus not a number but a random probability measure.

\section{The Markov property: when the present screens the past}\label{sec:markov}

In general, the conditional law of the future depends on the entire past. The Markov property is the assumption that makes the problem tractable: the present state contains all predictively relevant information. This is relative to the chosen filtration; a process can be Markov with respect to its natural filtration but fail to be Markov with respect to a larger one carrying additional predictive information.

\begin{definition}[Markov process]
An adapted process $X$ is Markov with respect to $(\mathcal{F}_t)_{t\ge0}$ if, for every bounded measurable $f$ and all $s,t\ge0$,
\begin{equation}\label{eq:markov}
\E\big[f(X_{t+s})\mid\mathcal{F}_t\big]=\E\big[f(X_{t+s})\mid X_t\big]=(P_sf)(X_t)\quad\text{a.s.},
\end{equation}
where $(P_s)_{s\ge0}$ is induced by transition kernels $p_s(x,\dd y)$ through $(P_sf)(x)=\int_E f(y)\,p_s(x,\dd y)$. The kernels are consistent through the Chapman--Kolmogorov equation,
\begin{equation}\label{eq:ck}
p_{t+s}(x,A)=\int_E p_s(y,A)\,p_t(x,\dd y),\qquad A\in\mathcal{E}.
\end{equation}
\end{definition}

To reach $A$ in time $t+s$, first propagate to an intermediate state $y$ in time $t$, then from $y$ to $A$ in time $s$, and integrate over all intermediate states. Equation \eqref{eq:markov} reduces ``the future given the past'' to ``the future given the present,'' the conceptual core of forward modeling \citep{EthierKurtz1986,RevuzYor1999}.

\begin{example}[Autoregression and state augmentation]
A scalar autoregression of order $k$ is not Markov in the scalar observation $X_t$ alone, but it becomes Markov after augmenting the state to $(X_t,X_{t-1},\dots,X_{t-k+1})$. This is the general modeling lesson: non-Markovian dynamics can often be represented as Markovian dynamics on a richer state space. The cost is dimensionality; the benefit is semigroup structure.
\end{example}

\begin{outhread}[the process and its Markov property]\label{ou:def}
Throughout, the worked example is the scalar Ornstein--Uhlenbeck process
\begin{equation}\label{eq:ousde}
\dd X_t=\theta(\mu-X_t)\dd t+\sigma\dd W_t,\qquad \theta>0,\ \sigma>0,
\end{equation}
with mean-reversion rate $\theta$, long-run level $\mu$, and volatility $\sigma$. It is a time-homogeneous Feller--Markov process with respect to the augmented natural filtration of $W$; everything that follows is an explicit instance of the corresponding abstract object.
\end{outhread}

\begin{cirthread}[a non-Gaussian, positivity-preserving counterpart]\label{cir:def}
The second worked example is the scalar Cox--Ingersoll--Ross process
\begin{equation}\label{eq:cirsde}
\dd X_t=\theta(\mu-X_t)\dd t+\sigma\sqrt{X_t}\,\dd W_t,\qquad \theta,\mu,\sigma>0,
\end{equation}
with the same mean-reversion rate, long-run level, and volatility parameters, but with \emph{state-dependent} (and degenerate) diffusion $\sigma^2 x$. On the state space $(0,\infty)$ it is a time-homogeneous Feller--Markov process \citep{CIR1985,Feller1951}. It shares the OU mean-reversion drift, so the two threads can be compared object by object, but its non-additive noise makes the transition law non-Gaussian and confines the process to the positive half-line --- the feature that makes it the standard model of a strictly positive quantity such as an interest rate or a variance.
\end{cirthread}

\bridge{The tower property of Proposition~\ref{prop:tower} \emph{is} the semigroup law. For a Markov process,
\[
(P_{t+s}f)(x)=\E_x[f(X_{t+s})]=\E_x\!\big[\E[f(X_{t+s})\mid\mathcal{F}_t]\big]=\E_x\!\big[(P_sf)(X_t)\big]=(P_t P_s f)(x),
\]
where the second equality is the tower property and the third is the Markov property \eqref{eq:markov}. Thus $P_{t+s}=P_tP_s$ below is not a new axiom but Proposition~\ref{prop:tower} specialized to the Markov filtration; equivalently, integrating the kernels reproduces Chapman--Kolmogorov \eqref{eq:ck}.}

\section{Semigroups: propagating observables and laws}\label{sec:semigroups}

Reading \eqref{eq:ck} at the level of operators turns transition kernels into a one-parameter semigroup acting on functions.

\begin{definition}[Transition semigroup]
The operators $(P_t)_{t\ge0}$ defined by $(P_tf)(x)=\E_x[f(X_t)]=\E[f(X_t)\mid X_0=x]$ form a semigroup:
\begin{equation}\label{eq:semigroup}
P_0=\Id,\qquad P_{t+s}=P_tP_s,\qquad s,t\ge0.
\end{equation}
On a suitable Banach space of functions, for example $C_0(E)$ for a Feller process, the semigroup is strongly continuous if $\norm{P_tf-f}\to0$ as $t\downarrow0$ for every $f$ in that space.
\end{definition}

The same kernel acts dually on probability measures. If $\mu$ is the law of $X_0$, then the law of $X_t$ is
\begin{equation}\label{eq:fwdmeasure}
\mu P_t(A)=\int_E p_t(x,A)\,\mu(\dd x),
\end{equation}
and duality gives
\begin{equation}\label{eq:duality}
\int_E f(y)\,(\mu P_t)(\dd y)=\int_E (P_tf)(x)\,\mu(\dd x).
\end{equation}
Thus $P_t$ propagates observables backward from the horizon to the present, while its adjoint propagates distributions forward from the present to the horizon. This duality is the operator-theoretic form of forecasting.

\begin{remark}[Observable view versus distribution view]
There are two equivalent ways to model the future. The first asks: given today's state $x$, what is the expected future value of every observable $f$? This is $P_tf(x)$. The second asks: given today's distribution $\mu$, what is tomorrow's distribution? This is $\mu P_t$. The backward Kolmogorov equation governs the first view; the forward equation governs the second.
\end{remark}

\begin{outhread}[transition law and spectral structure]\label{ou:semigroup}
Solving \eqref{eq:ousde} explicitly,
\begin{equation}\label{eq:ousol}
X_t=\mu+(X_0-\mu)\e^{-\theta t}+\sigma\!\int_0^t \e^{-\theta(t-s)}\dd W_s,
\end{equation}
so the transition law is Gaussian,
\[
X_t\mid X_0=x\ \sim\ \mathcal{N}\!\Big(m(t,x),\,v(t)\Big),\qquad
m(t,x)=\mu+(x-\mu)\e^{-\theta t},\quad v(t)=\frac{\sigma^2}{2\theta}\big(1-\e^{-2\theta t}\big).
\]
In particular $(P_t f)(x)=\int_\R f(y)\,\mathcal{N}\big(y;m(t,x),v(t)\big)\dd y$ and $(P_t\,\mathrm{id})(x)=\mu+(x-\mu)\e^{-\theta t}$. The OU semigroup has \emph{discrete spectrum}: the Hermite polynomials $\{H_n\}$ orthonormal for the invariant Gaussian are eigenfunctions,
\begin{equation}\label{eq:mehler}
P_t H_n=\e^{-n\theta t}H_n,\qquad n=0,1,2,\dots,
\end{equation}
which is Mehler's formula in spectral form \citep{BGL2014}. The decay rates $n\theta$ are exactly the relaxation times the rest of the thread will reproduce.
\end{outhread}

\begin{cirthread}[transition law and spectral structure]\label{cir:semigroup}
For CIR the conditional mean coincides with OU, $\E[X_t\mid X_0=x]=\mu+(x-\mu)\e^{-\theta t}$, but the conditional variance is state-dependent,
\[
\Var(X_t\mid X_0=x)=x\,\frac{\sigma^2}{\theta}\big(\e^{-\theta t}-\e^{-2\theta t}\big)+\mu\,\frac{\sigma^2}{2\theta}\big(1-\e^{-\theta t}\big)^2,
\]
and the transition law is a (scaled) \emph{noncentral chi-square}: with $c_t=2\theta/[\sigma^2(1-\e^{-\theta t})]$, the variable $2c_t X_t$ given $X_0=x$ is noncentral $\chi^2$ with $4\theta\mu/\sigma^2$ degrees of freedom and noncentrality $2c_t x\e^{-\theta t}$ \citep{CIR1985}. The semigroup again has discrete spectrum, but with the orthogonal family of the \emph{Gamma} invariant law: the generalized Laguerre polynomials $\{L_n^{(\alpha-1)}\}$, $\alpha=2\theta\mu/\sigma^2$, are eigenfunctions,
\begin{equation}\label{eq:laguerre}
P_t L_n^{(\alpha-1)}=\e^{-n\theta t}L_n^{(\alpha-1)},\qquad n=0,1,2,\dots,
\end{equation}
so the relaxation rates $n\theta$ match OU exactly while the eigenfunctions are Laguerre rather than Hermite \citep{KarlinTaylor1981,BGL2014}. OU and CIR are the Hermite and Laguerre members of the classical polynomial-eigenfunction diffusions.
\end{cirthread}

\section{Generators: the infinitesimal future}\label{sec:generators}

A strongly continuous semigroup is determined by its behavior at the origin.

\begin{definition}[Infinitesimal generator]
The generator of $(P_t)$ is the operator
\begin{equation}\label{eq:generator}
\mathcal{A}f=\lim_{t\downarrow0}\frac{P_tf-f}{t},
\end{equation}
defined on the domain $\mathcal{D}(\mathcal{A})$ of all $f$ for which the limit exists in the Banach-space norm.
\end{definition}

\begin{theorem}[Hille--Yosida; semigroup--generator correspondence]\label{thm:hilleyosida}
Let $\mathcal{A}$ be a densely defined closed operator on a Banach space $\mathcal{B}$. Then $\mathcal{A}$ generates a strongly continuous contraction semigroup $(P_t)_{t\ge0}$ if and only if $(0,\infty)\subseteq\rho(\mathcal{A})$ and
\[
\norm{(\lambda I-\mathcal{A})^{-n}}\le \lambda^{-n},
\qquad \lambda>0,\quad n\ge1.
\]
In that case, for every $f\in\mathcal{D}(\mathcal{A})$,
\begin{equation}\label{eq:gencommute}
\frac{\dd}{\dd t}P_tf=\mathcal{A}P_tf=P_t\mathcal{A}f.
\end{equation}
\end{theorem}

Formally one writes $P_t=\e^{t\mathcal{A}}$; this notation is literal when $\mathcal{A}$ is bounded and is understood in the semigroup sense otherwise. See \citet{Pazy1983} for the analytic theory and \citet{EthierKurtz1986} for the probabilistic formulation. Equation \eqref{eq:gencommute} says that the generator commutes with the flow it produces: on $\mathcal{D}(\mathcal{A})$ the square
\[
\begin{array}{ccc}
f & \xrightarrow{\ P_t\ } & P_tf\\[2pt]
{\scriptstyle\mathcal{A}}\big\downarrow & & \big\downarrow{\scriptstyle\mathcal{A}}\\[2pt]
\mathcal{A}f & \xrightarrow{\ P_t\ } & P_t\mathcal{A}f
\end{array}
\]
commutes. The generator is the local law of motion; the semigroup is its time-integrated law. Everything operational about how Markovian futures evolve is encoded in $(\mathcal{A},\mathcal{D}(\mathcal{A}))$.

\begin{remark}[Why the domain matters]
The pair $(\mathcal{A},\mathcal{D}(\mathcal{A}))$ matters more than the formal expression for $\mathcal{A}$. Different boundary conditions can give the same differential expression but different domains, hence different processes: absorbing, reflecting, and killed diffusions are distinguished by their generator domains. A model of the future is incomplete until the domain and boundary behavior are specified.
\end{remark}

\begin{remark}[Local and nonlocal generators]
For a diffusion the generator is a second-order \emph{local} differential operator \eqref{eq:diffgen}. This is not the only possibility: a jump or L\'evy-type process has a generator with an additional integral term $\int\big(f(x+z)-f(x)-\mathbf{1}_{\{|z|\le1\}}z\cdot\nabla f(x)\big)\nu(x,\dd z)$, so ``local law of motion'' must be read as ``local plus jump'' in general \citep{Applebaum2009}. The diffusion case is the one carried by the OU thread.
\end{remark}

\begin{outhread}[generator and spectrum]\label{ou:gen}
For \eqref{eq:ousde} the generator on $C_c^2(\R)$ is
\[
\mathcal{A}f(x)=\theta(\mu-x)f'(x)+\tfrac{1}{2}\sigma^2 f''(x).
\]
Its Hermite eigenfunctions satisfy $\mathcal{A}H_n=-n\theta\,H_n$, consistent with \eqref{eq:mehler} via $P_t=\e^{t\mathcal{A}}$. The spectral gap $\theta$ (the smallest nonzero relaxation rate) will reappear as the convergence rate of the law (\S\ref{sec:wasserstein}) and of time averages (\S\ref{sec:farfuture}).
\end{outhread}

\begin{cirthread}[generator, and the boundary remark made concrete]\label{cir:gen}
For \eqref{eq:cirsde} the generator on a suitable core in $C^2((0,\infty))$ is
\[
\mathcal{A}f(x)=\theta(\mu-x)f'(x)+\tfrac{1}{2}\sigma^2 x\,f''(x),
\]
with the same Laguerre spectrum $\mathcal{A}L_n^{(\alpha-1)}=-n\theta L_n^{(\alpha-1)}$. Because the diffusion coefficient $a(x)=\sigma^2 x$ \emph{degenerates at the boundary} $x=0$, the abstract domain-and-boundary discussion above becomes an explicit dichotomy: the Feller condition $2\theta\mu\ge\sigma^2$ makes $0$ inaccessible and the process stays strictly positive, whereas if $2\theta\mu<\sigma^2$ the boundary is attainable, and the precise Markov process is determined by the Feller boundary classification and the generator domain \citep{Feller1951}. Alternative boundary prescriptions correspond to different extensions of the same formal differential operator. This is the boundary remark of \S\ref{sec:generators} realized: the differential expression alone does not determine the process; the boundary classification does.
\end{cirthread}

\section{From dynamics to PDEs: Kolmogorov equations}\label{sec:kolmogorov}

Specializing \eqref{eq:gencommute} to $u(t,x)=(P_tf)(x)=\E_x[f(X_t)]$ yields the equation governing conditional expectations of future observables.

\begin{theorem}[Backward Kolmogorov equation]\label{thm:backward}
Let $(P_t)$ be a strongly continuous Markov semigroup with generator $\mathcal{A}$. If $f\in\mathcal{D}(\mathcal{A})$ and the required regularity holds, then $u(t,x)=P_tf(x)$ solves
\begin{equation}\label{eq:backward}
\partial_t u(t,x)=(\mathcal{A}u)(t,x),\qquad u(0,x)=f(x).
\end{equation}
\end{theorem}

For a diffusion solving
\begin{equation}\label{eq:diffsde}
\dd X_t=b(X_t)\dd t+\sigma(X_t)\dd W_t,
\end{equation}
with $W$ an $m$-dimensional Brownian motion, $b:\R^d\to\R^d$, and $\sigma:\R^d\to\R^{d\times m}$, the diffusion matrix is $a(x)=\sigma(x)\sigma(x)^\top$ and the generator is the second-order diffusion operator, elliptic when $a(x)$ is nondegenerate,
\begin{equation}\label{eq:diffgen}
\mathcal{A}f(x)=\sum_{i=1}^d b_i(x)\partial_i f(x)+\tfrac{1}{2}\sum_{i,j=1}^d a_{ij}(x)\partial^2_{ij}f(x).
\end{equation}
Its formal $L^2$-adjoint $\mathcal{A}^*$ governs the evolution of densities.

\begin{theorem}[Forward Kolmogorov / Fokker--Planck equation]\label{thm:forward}
Assume $X_t$ admits a density $p(t,\cdot)$ and that the coefficients and density are sufficiently regular. Then
\begin{equation}\label{eq:forward}
\partial_t p(t,y)=(\mathcal{A}^* p)(t,y)=-\sum_{i=1}^d\partial_i\big(b_i(y)p(t,y)\big)+\tfrac{1}{2}\sum_{i,j=1}^d\partial^2_{ij}\big(a_{ij}(y)p(t,y)\big).
\end{equation}
\end{theorem}

\bridge{The two equations come from two one-line calculations, and both are derivations rather than analogies. \emph{Backward:} take $\E_x$ of It\^o's formula \eqref{eq:ito} (the martingale term has zero mean), giving $P_tf(x)-f(x)=\int_0^t P_s(\mathcal{A}f)(x)\dd s$; differentiating in $t$ yields \eqref{eq:gencommute} and hence \eqref{eq:backward}. \emph{Forward:} $\mathcal{A}^*$ is defined by $\inner{\mathcal{A}f}{g}_{L^2}=\inner{f}{\mathcal{A}^*g}_{L^2}$; in one dimension, integrating by parts (twice for the second-order term, discarding boundary terms),
\[
\int\!\Big(b f'+\tfrac{a}{2}f''\Big)g\dd x=\int f\Big(-(bg)'+\tfrac{1}{2}(a g)''\Big)\dd x,
\]
so $\mathcal{A}^*g=-(bg)'+\tfrac12(ag)''$, which is \eqref{eq:forward}. The backward and forward equations are therefore literally adjoint, and differentiating the duality \eqref{eq:duality} states it cleanly:
\[
\frac{\dd}{\dd t}\int f(y)p(t,y)\dd y=\int(\mathcal{A}f)(y)p(t,y)\dd y=\int f(y)(\mathcal{A}^*p)(t,y)\dd y.
\]}

Adding a killing or discount potential gives the Feynman--Kac formula. If $V\ge0$ and
\begin{equation}\label{eq:fk}
u(t,x)=\E_x\!\Big[f(X_t)\exp\!\Big(-\!\int_0^t V(X_s)\dd s\Big)\Big],
\end{equation}
then, under standard conditions, $\partial_t u=\mathcal{A}u-Vu$ with $u(0,\cdot)=f$. This is the bridge between expectations of path functionals and linear PDEs \citep{KaratzasShreve1991,Oksendal2003}, and it is the form that returns in pricing (\S\ref{sec:pricing}).

\begin{outhread}[both Kolmogorov equations and the invariant law]\label{ou:kolmogorov}
For OU the backward and forward operators are
\[
\mathcal{A}u=\theta(\mu-x)\partial_x u+\tfrac{\sigma^2}{2}\partial_{xx}u,\qquad
\mathcal{A}^*p=-\partial_y\big(\theta(\mu-y)p\big)+\tfrac{\sigma^2}{2}\partial_{yy}p.
\]
The stationary forward equation $\mathcal{A}^*p_\infty=0$ has the Gaussian solution
\begin{equation}\label{eq:ouinvariant}
p_\infty(y)\ \propto\ \exp\!\Big(-\frac{\theta(y-\mu)^2}{\sigma^2}\Big),\qquad\text{i.e.}\quad p_\infty=\mathcal{N}\!\Big(\mu,\frac{\sigma^2}{2\theta}\Big),
\end{equation}
the $t\to\infty$ limit of $v(t)$ in Thread~\ref{ou:semigroup}. The invariant law is read off as the kernel of $\mathcal{A}^*$.
\end{outhread}

\begin{cirthread}[forward operator, Gamma invariant law, and ergodicity]\label{cir:kolmogorov}
For CIR the forward operator is
\[
\mathcal{A}^*p=-\partial_y\big(\theta(\mu-y)p\big)+\tfrac{\sigma^2}{2}\partial_{yy}\big(y\,p\big),
\]
and the stationary equation $\mathcal{A}^*p_\infty=0$ has the \emph{Gamma} solution
\begin{equation}\label{eq:cirinvariant}
p_\infty(y)\ \propto\ y^{\alpha-1}\e^{-\beta y}\ \ (y>0),\qquad \alpha=\frac{2\theta\mu}{\sigma^2},\quad \beta=\frac{2\theta}{\sigma^2},
\end{equation}
i.e.\ $p_\infty=\mathrm{Gamma}(\alpha,\beta)$ with mean $\alpha/\beta=\mu$ and variance $\alpha/\beta^2=\mu\sigma^2/2\theta$, the $t\to\infty$ limit of the conditional variance in Thread~\ref{cir:semigroup}. Where OU relaxes to a Gaussian, CIR relaxes to a Gamma. The process is positive Harris recurrent and exponentially ergodic at the spectral-gap rate $\theta$ \citep{MeynTweedie2009,BGL2014}, the non-Gaussian analogue of the $W_2$ contraction of \S\ref{sec:wasserstein}.
\end{cirthread}

\section{The geometry of forward evolution: a Wasserstein gradient flow}\label{sec:wasserstein}

The forward equation \eqref{eq:forward} is an evolution of \emph{densities}; it therefore has a geometry, and that geometry is the natural meaning of ``the forward evolution of laws.'' Let $\mathcal{P}_2(\R^d)$ be the probability measures with finite second moment, equipped with the quadratic Wasserstein distance
\[
W_2(\rho_0,\rho_1)^2=\inf_{\gamma\in\Pi(\rho_0,\rho_1)}\int_{\R^d\times\R^d}|x-y|^2\dd\gamma(x,y),
\]
the infimum over couplings $\gamma$ with marginals $\rho_0,\rho_1$ \citep{Villani2009}. For a reversible diffusion with gradient drift $b=-\nabla\Psi$ and constant diffusion $D=\tfrac12\sigma^2$, the Fokker--Planck equation reads $\partial_t\rho=\nabla\!\cdot(\rho\nabla\Psi)+D\Delta\rho$.

\begin{theorem}[Jordan--Kinderlehrer--Otto]\label{thm:jko}
The Fokker--Planck flow $\partial_t\rho=\nabla\!\cdot(\rho\nabla\Psi)+D\Delta\rho$ is the gradient flow in $(\mathcal{P}_2,W_2)$ of the free energy
\[
\mathcal{F}(\rho)=\int_{\R^d}\Psi\,\rho\dd x+D\!\int_{\R^d}\rho\log\rho\dd x
=D\,\KL(\rho\,\|\,\rho_\infty)+\text{const},\qquad \rho_\infty\propto \e^{-\Psi/D}.
\]
It is realized by the minimizing-movement (JKO) scheme: with step $\tau>0$,
\begin{equation}\label{eq:jko}
\rho^{(k+1)}=\argmin_{\rho\in\mathcal{P}_2}\Big\{\mathcal{F}(\rho)+\frac{1}{2\tau}W_2(\rho,\rho^{(k)})^2\Big\},
\end{equation}
which converges to the continuous flow as $\tau\downarrow0$ \citep{JKO1998,Otto2001,AGS2008}.
\end{theorem}

The forward evolution of laws is thus a steepest descent of free energy: the system moves to lower energy $\int\Psi\rho$ and higher entropy, and the two balance at the Gibbs measure $\rho_\infty$. With the convention
\[
\mathcal{I}(\rho\|\rho_\infty)=\int_{\R^d}\left|\nabla\log\frac{\rho}{\rho_\infty}\right|^2\rho\dd x,
\]
energy dissipation is exactly the decay of relative entropy,
\[
\frac{\dd}{\dd t}\KL(\rho_t\|\rho_\infty)=-D\,\mathcal{I}(\rho_t\|\rho_\infty)\le0,
\]
where a different normalization of $\mathcal{I}$ merely moves the factor $D$ \citep{BGL2014}.

\begin{outhread}[mean reversion as Wasserstein contraction]\label{ou:wasserstein}
OU is the canonical instance: $\Psi(x)=\tfrac{\theta}{2}(x-\mu)^2$ is $\theta$-convex and $D=\tfrac12\sigma^2$, so $\rho_\infty=\mathcal{N}(\mu,\sigma^2/2\theta)$ as in \eqref{eq:ouinvariant} and the flow is uniformly geodesically convex. The contraction rate is exactly the spectral gap $\theta$, by an elementary synchronous coupling: if $X$ and $X'$ solve \eqref{eq:ousde} driven by the \emph{same} $W$, then $\dd(X_t-X'_t)=-\theta(X_t-X'_t)\dd t$, hence $|X_t-X'_t|=|X_0-X'_0|\e^{-\theta t}$ pathwise, and therefore
\begin{equation}\label{eq:w2contract}
W_2(\rho_t,\rho'_t)\le \e^{-\theta t}\,W_2(\rho_0,\rho'_0).
\end{equation}
Mean reversion, the relaxation rate $n\theta$ of \eqref{eq:mehler}, and exponential convergence to the Gaussian invariant law are one and the same fact, seen as a contraction in $W_2$.
\end{outhread}

Figure~\ref{fig:ou} makes the thread concrete: every curve is evaluated from a closed form derived above, with no simulation. Panel~(a) shows the law \eqref{eq:ousol} relaxing to the invariant Gaussian \eqref{eq:ouinvariant}; panel~(b) the mean reversion of \eqref{eq:ousol} with the $\pm\sqrt{v(t)}$ band; panel~(c) the spectral decay \eqref{eq:mehler}; and panel~(d) anticipates \S\ref{sec:pricing}, the Vasicek yield curves implied by the same generator.

\begin{figure}[htbp]
\centering
\includegraphics[width=\textwidth]{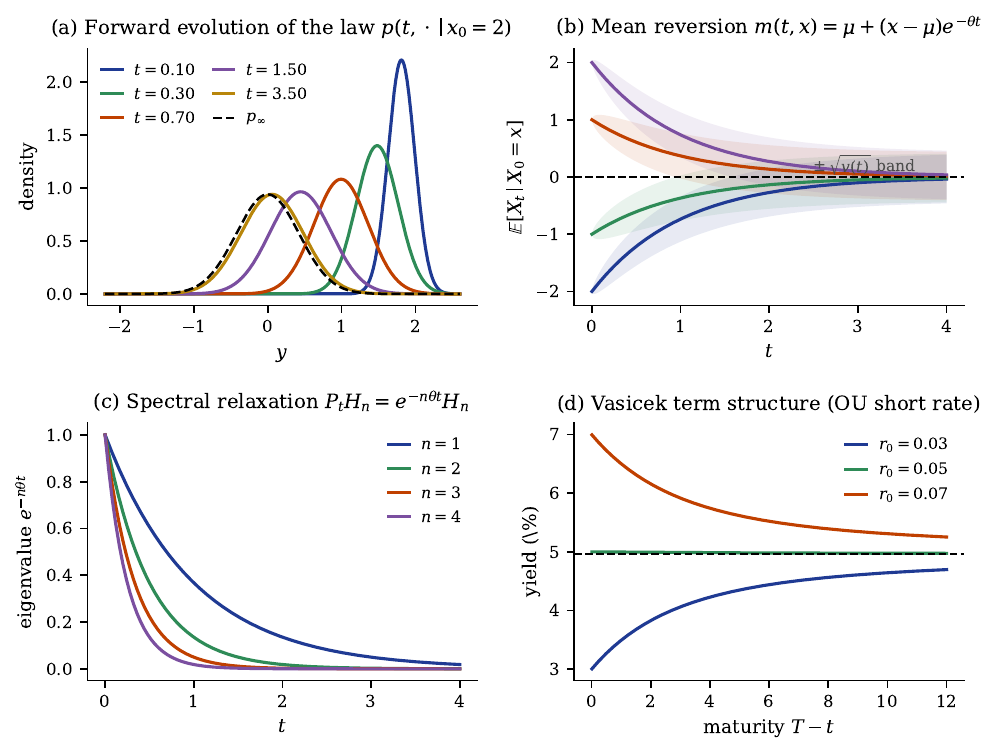}
\caption{The Ornstein--Uhlenbeck thread, computed from the paper's closed forms ($\theta=1$, $\mu=0$, $\sigma=0.6$; panel (d) uses short-rate parameters $\theta=0.6,\mu=0.05,\sigma=0.015$). (a) The conditional density $p(t,\cdot\mid x_0)$ evolves forward to the invariant law $p_\infty=\mathcal{N}(\mu,\sigma^2/2\theta)$. (b) The conditional mean reverts at rate $\theta$ while the conditional variance saturates at $\sigma^2/2\theta$. (c) The Hermite eigenvalues $\e^{-n\theta t}$ set the relaxation times. (d) The implied Vasicek term structure is upward-, flat-, or downward-sloping as the short rate sits below, at, or above its long-run level.}
\label{fig:ou}
\end{figure}

\subsection{Three geometries: Hilbert, Fisher--Rao, Wasserstein}\label{sec:threegeom}
The Wasserstein picture is one corner of a triangle of geometries that the rest of the paper has been using implicitly, each natural to a different object of forecasting.
\begin{itemize}\itemsep2pt
\item \emph{Hilbert ($L^2$) geometry of observables.} Conditional expectation is orthogonal projection in $L^2(\PP)$ (Theorem~\ref{thm:projection}); point forecasting is a Hilbert-space projection and the residual is orthogonal to present information. This is the geometry of \S\ref{sec:conditional}.
\item \emph{Fisher--Rao geometry of model families.} On a parametric family $\{p_\vartheta\}$ the Fisher information $g_{ij}(\vartheta)=\E_\vartheta[\partial_i\log p_\vartheta\,\partial_j\log p_\vartheta]$ is the canonical Riemannian metric, with relative entropy as the associated divergence \citep{Amari2016,AyJostLeSchwachhoefer2017}. This is the geometry in which the \emph{estimation} of \S\ref{sec:estimation} lives, and the Cram\'er--Rao bound is its length element.
\item \emph{Wasserstein geometry of laws.} On $\mathcal{P}_2$ the metric $W_2$ is the one under which the forward equation is the gradient flow of free energy (Theorem~\ref{thm:jko}). This is the geometry of \S\ref{sec:wasserstein}.
\end{itemize}
The three meet rather than merely coexist. The same relative Fisher information $\mathcal{I}(\rho\|\rho_\infty)$ that controls Wasserstein-side entropy dissipation, $\tfrac{\dd}{\dd t}\KL=-D\,\mathcal{I}$, is the infinitesimal Fisher--Rao length along the flow; and the Gaussian family of the OU thread carries explicit forms of all three metrics at once, just as the Gamma family does for CIR. Forecasting (Hilbert), estimating the generator (Fisher--Rao), and propagating the law (Wasserstein) are thus three faces of one geometric object.

\section{The pathwise picture: stochastic differential equations}\label{sec:sde}

Generators describe the future in law; SDEs describe it path by path.

\begin{theorem}[Existence, uniqueness, Markov property, generator]\label{thm:sde}
Assume $b:\R^d\to\R^d$ and $\sigma:\R^d\to\R^{d\times m}$ are globally Lipschitz with linear growth. Then
\begin{equation}\label{eq:sde}
\dd X_t=b(X_t)\dd t+\sigma(X_t)\dd W_t,\qquad X_0=x,
\end{equation}
has a unique strong solution; the solution is a Feller--Markov process with generator \eqref{eq:diffgen} on $C^2_c(\R^d)$. For $f\in C^2_c(\R^d)$, It\^o's formula gives
\begin{equation}\label{eq:ito}
\dd f(X_t)=\underbrace{(\mathcal{A}f)(X_t)\dd t}_{\text{predictable drift}}+\underbrace{\nabla f(X_t)^\top\sigma(X_t)\dd W_t}_{\text{martingale increment}}.
\end{equation}
\end{theorem}

Equation \eqref{eq:ito} is the stochastic chain rule for the future: the generator gives the predictable instantaneous change of every smooth observable; the It\^o integral is the irreducible innovation. As shown in the Bridge of \S\ref{sec:kolmogorov}, taking expectations recovers the backward equation. See \citet{KaratzasShreve1991}, \citet{Oksendal2003}, and \citet{StroockVaradhan1979}.

\begin{definition}[Martingale problem]
Let $\mathcal{A}$ be an operator on a domain $\mathcal{D}(\mathcal{A})$. A process $X$ solves the martingale problem for $(\mathcal{A},\mathcal{D}(\mathcal{A}))$ if, for every $f\in\mathcal{D}(\mathcal{A})$,
\begin{equation}\label{eq:mgproblem}
M^f_t=f(X_t)-f(X_0)-\int_0^t(\mathcal{A}f)(X_s)\dd s
\end{equation}
is a martingale.
\end{definition}

The martingale problem reverses the SDE logic: instead of starting with paths and deriving the generator, it starts with the generator and asks whether a process exists whose observables have exactly the prescribed drift. Well-posedness is equivalent, in many settings, to uniqueness in law of the corresponding Markov process \citep{StroockVaradhan1979,EthierKurtz1986}.

\begin{outhread}[explicit drift--innovation split]\label{ou:sde}
Solution \eqref{eq:ousol} is already in predictable-plus-martingale form: $\mu+(X_0-\mu)\e^{-\theta t}$ is the predictable finite-variation part and $\sigma\int_0^t\e^{-\theta(t-s)}\dd W_s$ is the Gaussian martingale innovation, with conditional variance $v(t)$. The martingale problem \eqref{eq:mgproblem} is solved explicitly: $f(X_t)-f(X_0)-\int_0^t\mathcal{A}f(X_s)\dd s=\int_0^t \sigma f'(X_s)\dd W_s$.
\end{outhread}

\section{Martingales: the unpredictable part}\label{sec:martingales}

A forecast that already incorporates all available information should not be systematically improvable. This is the martingale property.

\begin{definition}[Martingale]
An adapted integrable process $M$ is a martingale if $\E[M_t\mid\mathcal{F}_s]=M_s$ for $0\le s\le t$. It is a local martingale if there exist stopping times $\tau_n\uparrow\infty$ such that $M^{\tau_n}$ is a martingale for every $n$.
\end{definition}

A martingale is a fair game: its best forecast of any future value is its present value. Decomposition theorems make precise the idea that every sufficiently regular future is predictable trend plus unpredictable innovation.

\begin{theorem}[Doob--Meyer decomposition]\label{thm:doobmeyer}
Every right-continuous submartingale $X$ of class (D) admits a unique decomposition
\begin{equation}\label{eq:doobmeyer}
X_t=X_0+A_t+M_t,
\end{equation}
where $A$ is predictable, increasing, integrable, $A_0=0$, and $M$ is a uniformly integrable martingale.
\end{theorem}

The predictable part $A$ is systematic accumulation; the martingale part $M$ is residual surprise. In Brownian filtrations the surprise has a canonical representation.

\begin{theorem}[Martingale representation]\label{thm:mgrep}
Let $(\mathcal{F}_t)$ be the augmented natural filtration of an $m$-dimensional Brownian motion $W$. Every square-integrable $(\mathcal{F}_t)$-martingale $M$ admits
\begin{equation}\label{eq:mgrep}
M_t=M_0+\int_0^t H_s\dd W_s,
\end{equation}
for a predictable $H$ with $\E\int_0^T\norm{H_s}^2\dd s<\infty$ for each finite $T$.
\end{theorem}

In a Brownian world, every uncertain future is the present plus accumulated fresh noise weighted by a predictable sensitivity. This underlies hedging in complete markets, innovations in filtering, and stochastic control \citep{RevuzYor1999,KaratzasShreve1991,Protter2005}. The continuous-time decomposition \eqref{eq:doobmeyer} has an exact discrete analogue (\S\ref{sec:discrete}).

\begin{remark}[The integrand is a conditional expectation]
Martingale representation \eqref{eq:mgrep} is existential, but the Clark--Ocone formula makes the integrand explicit: for $F=M_T$ smooth in the Malliavin sense, $H_t=\E[D_tF\mid\mathcal{F}_t]$, the predictable projection of the Malliavin derivative \citep{Nualart2006}. The unpredictable future is thus reconstructed from a \emph{conditional expectation} of an infinitesimal perturbation --- the same projection that defines forecasting (\S\ref{sec:conditional}) and filtering (\S\ref{sec:filtering}), and, read financially, the hedge ratio.
\end{remark}

\section{Filtering: the future seen through noise}\label{sec:filtering}

Often the state $X$ is not observed directly. One observes a noisy process $Y$ and must forecast the hidden signal using the observation filtration $\mathcal{Y}_t=\sigma(Y_s:s\le t)$.

\begin{definition}[Filter]
For a signal $X$ and observation filtration $(\mathcal{Y}_t)$, the nonlinear filter is the measure-valued process $\pi_t(f)=\E[f(X_t)\mid\mathcal{Y}_t]$, defined for bounded measurable $f$.
\end{definition}

By Theorem~\ref{thm:projection}, $\pi_t(\mathrm{id})$ is the minimum-variance estimate of the current state from observations. The filter itself evolves by a Kushner--Stratonovich equation: for the diffusion signal with generator $\mathcal{A}$ and scalar observation $\dd Y_t=h(X_t)\dd t+\dd V_t$,
\begin{equation}\label{eq:ks}
\dd\pi_t(f)=\pi_t(\mathcal{A}f)\dd t+\big(\pi_t(hf)-\pi_t(h)\pi_t(f)\big)\big(\dd Y_t-\pi_t(h)\dd t\big),
\end{equation}
the forward (generator) term plus a correction driven by the innovation \citep{BainCrisan2009,LiptserShiryaev2001}. In the linear--Gaussian case the filter is finite-dimensional and explicit.

\begin{theorem}[Kalman--Bucy filter]\label{thm:kalman}
Consider the linear system
\[
\dd X_t=A X_t\dd t+B\dd W_t,\qquad \dd Y_t=C X_t\dd t+D\dd V_t,
\]
with $W,V$ independent Brownian motions and $DD^\top\succ0$. If the initial state is Gaussian and independent of $(W,V)$, then the conditional law of $X_t$ given $\mathcal{Y}_t$ is Gaussian with mean $\widehat{X}_t=\E[X_t\mid\mathcal{Y}_t]$ and covariance $\Sigma_t$ solving
\begin{align}
\dd\widehat{X}_t&=A\widehat{X}_t\dd t+\Sigma_t C^\top(DD^\top)^{-1}\big(\dd Y_t-C\widehat{X}_t\dd t\big),\label{eq:kbmean}\\
\dot{\Sigma}_t&=A\Sigma_t+\Sigma_t A^\top+BB^\top-\Sigma_t C^\top(DD^\top)^{-1}C\Sigma_t.\label{eq:kbric}
\end{align}
\end{theorem}

The process $\dd\nu_t=\dd Y_t-C\widehat{X}_t\dd t$ is the innovations process: the part of each new observation not predicted by the current filter. After normalization by the observation covariance it is a Brownian motion in $(\mathcal{Y}_t)$.

\bridge{The innovation is a martingale by construction. Since $C\widehat{X}_t=\E[CX_t\mid\mathcal{Y}_t]$, for $s<t$
\[
\E[\nu_t-\nu_s\mid\mathcal{Y}_s]=\E\!\Big[\int_s^t\big(CX_u-C\widehat{X}_u\big)\dd u+\int_s^t D\dd V_u\,\Big|\,\mathcal{Y}_s\Big]=0,
\]
because $\E[CX_u-C\widehat{X}_u\mid\mathcal{Y}_s]=0$ for $u\ge s$ by the tower property and the observation noise is mean-zero and independent. Filtering is thus the recursive conversion of observations into a conditional law, and innovations are exactly the martingale surprises of \S\ref{sec:martingales} left after the best forecast is removed.}

\begin{outhread}[the canonical linear--Gaussian signal]\label{ou:filter}
Taking the signal to be OU, $A=-\theta$ and $B=\sigma$ in Theorem~\ref{thm:kalman}: the mean-reverting process \emph{is} the textbook Kalman--Bucy signal, and the scalar Riccati equation \eqref{eq:kbric} has a stationary solution $\Sigma_\infty$ giving the steady-state filter gain. The same Riccati equation reappears, transposed, in control (\S\ref{sec:control}).
\end{outhread}

\section{Pricing the future: martingale measures}\label{sec:pricing}

In financial markets the relevant future is the discounted value of payoffs. Absence of arbitrage forces the existence of measures under which discounted prices are martingales.

\begin{theorem}[Fundamental theorem of asset pricing]\label{thm:ftap}
Let $S$ be a locally bounded semimartingale price process. Then $S$ satisfies no free lunch with vanishing risk if and only if there exists a probability measure $\QQ\sim\PP$ under which the discounted price process is a local martingale.
\end{theorem}

Under such a $\QQ$, the no-arbitrage value of a claim paying $V_T$ at maturity $T$ is
\begin{equation}\label{eq:price}
V_t=\theta_t^{-1}\,\E^{\QQ}\!\big[\theta_T V_T\mid\mathcal{F}_t\big],\qquad \theta_t=\exp\!\Big(-\!\int_0^t r_s\dd s\Big),
\end{equation}
so the discounted value $\theta_t V_t$ is a $\QQ$-martingale. The pricing operator is a conditional expectation, but under a risk-neutral measure rather than the physical one. In complete markets this measure is unique; in incomplete markets the equivalent martingale measure need not be unique, so \eqref{eq:price} is the price associated with the chosen pricing measure while no-arbitrage alone may determine only a range of admissible prices. In Markov diffusion models \eqref{eq:price} becomes a Feynman--Kac representation \eqref{eq:fk} and hence a backward pricing PDE \citep{DelbaenSchachermayer2006,KaratzasShreve1991,Duffie2001,Bjork2009}.

\begin{remark}[Physical versus pricing futures]
The physical measure $\PP$ describes empirical frequencies; the pricing measure $\QQ$ encodes no-arbitrage valuation. A physical forecast asks what is likely to happen; a no-arbitrage price asks which present value prevents a self-financing strategy from generating arbitrage. The two coincide only under special preferences or market-price-of-risk structures, and the gap between them is itself a modeling object (\S\ref{sec:estimation}).
\end{remark}

\begin{outhread}[OU as Vasicek: Feynman--Kac becomes an affine term structure]\label{ou:pricing}
Let the short rate follow OU under $\QQ$, $\dd r_t=\theta(\mu-r_t)\dd t+\sigma\dd W_t$. By \eqref{eq:price} the zero-coupon bond price $P(t,T)=\E^{\QQ}[\exp(-\int_t^T r_u\dd u)\mid\mathcal{F}_t]$ solves the Feynman--Kac PDE
\[
\partial_t P+\theta(\mu-r)\partial_r P+\tfrac{\sigma^2}{2}\partial_{rr}P-rP=0,\qquad P(T,T)=1.
\]
The affine ansatz $P(t,T)=\exp\!\big(\mathfrak{a}(t,T)-\mathfrak{b}(t,T)\,r_t\big)$ separates it into ODEs with the closed-form Vasicek solution
\begin{align*}
\mathfrak{b}(t,T)&=\frac{1-\e^{-\theta(T-t)}}{\theta},\\
\mathfrak{a}(t,T)&=\Big(\mu-\frac{\sigma^2}{2\theta^2}\Big)\big(\mathfrak{b}(t,T)-(T-t)\big)-\frac{\sigma^2}{4\theta}\,\mathfrak{b}(t,T)^2,
\end{align*}
\citep{Vasicek1977,Bjork2009}. The abstract pricing identity \eqref{eq:price} has become an explicit yield curve, and the operator/spectral view of long-horizon discounting is the subject of \citet{HansenScheinkman2009}.
\end{outhread}

\begin{cirthread}[CIR as a positive short rate: a genuine Riccati term structure]\label{cir:pricing}
Let the short rate follow CIR under $\QQ$, $\dd r_t=\theta(\mu-r_t)\dd t+\sigma\sqrt{r_t}\,\dd W_t$, which (unlike Vasicek) keeps rates nonnegative. The bond price $P(t,T)=\E^{\QQ}[\exp(-\int_t^T r_u\dd u)\mid\mathcal{F}_t]$ solves
\[
\partial_t P+\theta(\mu-r)\partial_r P+\tfrac{\sigma^2}{2}r\,\partial_{rr}P-rP=0,\qquad P(T,T)=1,
\]
and the affine ansatz $P=\exp(\mathfrak{a}(\tau)-\mathfrak{b}(\tau)r)$, $\tau=T-t$, now separates into
\begin{equation}\label{eq:cirriccati}
\mathfrak{b}'=1-\theta\mathfrak{b}-\tfrac{1}{2}\sigma^2\mathfrak{b}^2,\qquad \mathfrak{a}'=-\theta\mu\,\mathfrak{b},
\end{equation}
where the $\mathfrak{b}$-equation is a \emph{genuine Riccati} equation --- the quadratic term is the trace of the $\sqrt{r}$ diffusion --- in contrast to the \emph{linear} Vasicek $\mathfrak{b}$-equation. With $\gamma=\sqrt{\theta^2+2\sigma^2}$ its solution is, in closed form \citep{CIR1985},
\[
\mathfrak{b}(\tau)=\frac{2(\e^{\gamma\tau}-1)}{(\gamma+\theta)(\e^{\gamma\tau}-1)+2\gamma},\qquad
\mathfrak{a}(\tau)=\frac{2\theta\mu}{\sigma^2}\,\log\!\frac{2\gamma\,\e^{(\gamma+\theta)\tau/2}}{(\gamma+\theta)(\e^{\gamma\tau}-1)+2\gamma}.
\]
The Riccati structure that appeared in filtering \eqref{eq:kbric} and control (Thread~\ref{ou:control}) thus reappears a third time, here in the term structure itself.
\end{cirthread}

\begin{table}[tbp]
\centering
\renewcommand{\arraystretch}{1.32}
\begin{tabular}{@{}L{0.26\textwidth} L{0.34\textwidth} L{0.34\textwidth}@{}}
\toprule
\textbf{Object} & \textbf{Ornstein--Uhlenbeck} & \textbf{Cox--Ingersoll--Ross}\\
\midrule
SDE & $\dd X=\theta(\mu-X)\dd t+\sigma\dd W$ & $\dd X=\theta(\mu-X)\dd t+\sigma\sqrt{X}\,\dd W$\\
State space & $\R$ & $(0,\infty)$\\
Generator & $\theta(\mu-x)\partial_x+\tfrac12\sigma^2\partial_{xx}$ & $\theta(\mu-x)\partial_x+\tfrac12\sigma^2x\,\partial_{xx}$\\
Conditional mean & $\mu+(x-\mu)\e^{-\theta t}$ & $\mu+(x-\mu)\e^{-\theta t}$ (same)\\
Transition law & Gaussian & noncentral $\chi^2$\\
Invariant law & $\mathcal{N}(\mu,\sigma^2/2\theta)$ & $\mathrm{Gamma}(2\theta\mu/\sigma^2,\,2\theta/\sigma^2)$\\
Eigenfunctions & Hermite $H_n$ & Laguerre $L_n^{(\alpha-1)}$\\
Eigenvalues & $-n\theta$ & $-n\theta$ (same)\\
Boundary at $0$ & not applicable & accessible iff $2\theta\mu<\sigma^2$ (Feller)\\
Zero-coupon bond & affine, linear $\mathfrak{b}$-ODE (Vasicek) & affine, Riccati $\mathfrak{b}$-ODE (CIR)\\
\bottomrule
\end{tabular}
\caption{Two complete threads on the same spine. The shared drift gives identical mean reversion and identical relaxation rates $-n\theta$; the additive-versus-multiplicative noise gives Gaussian-versus-Gamma stationarity, Hermite-versus-Laguerre spectra, an unconstrained-versus-half-line state space, and a linear-versus-Riccati term structure. The framework is visibly not Gaussian-specific.}
\label{tab:oucir}
\end{table}

\begin{figure}[htbp]
\centering
\includegraphics[width=\textwidth]{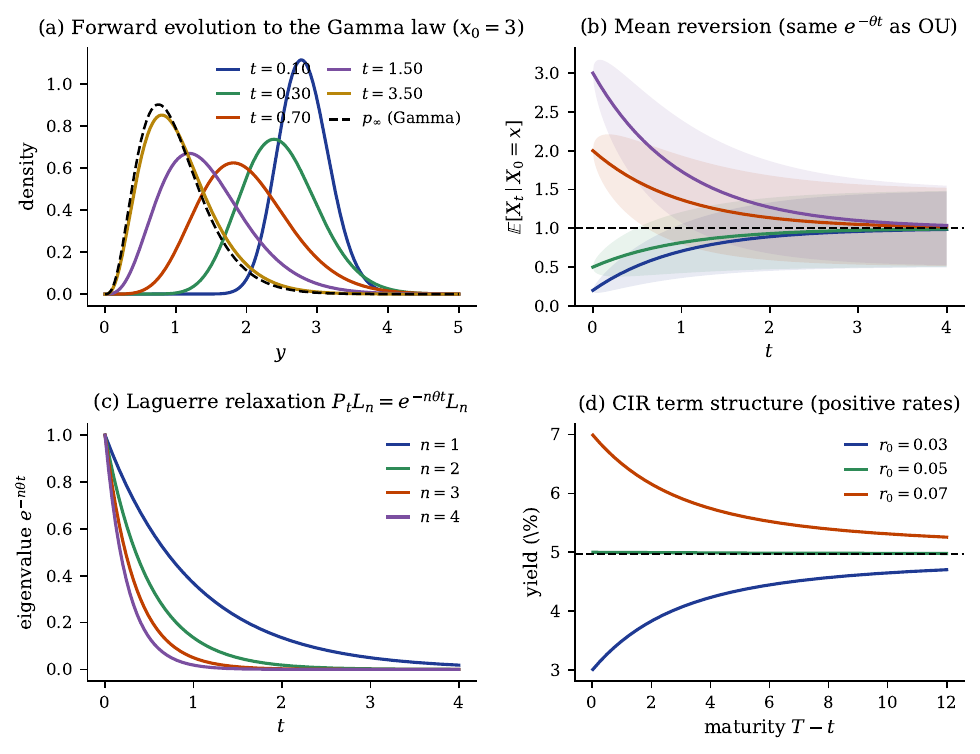}
\caption{The Cox--Ingersoll--Ross thread, the non-Gaussian counterpart of Figure~\ref{fig:ou}, computed from closed forms ($\theta=1$, $\mu=1$, $\sigma=0.7$, Feller condition $2\theta\mu\ge\sigma^2$ satisfied; panel (d) uses short-rate parameters $\theta=0.6,\mu=0.05,\sigma=0.06$). (a) The conditional density (noncentral $\chi^2$) evolves forward to the \emph{Gamma} invariant law, skewed and supported on $(0,\infty)$. (b) The conditional mean reverts at rate $\theta$ exactly as for OU, while the conditional variance is state-dependent. (c) The Laguerre eigenvalues $\e^{-n\theta t}$ give the same relaxation rates as OU. (d) The implied CIR term structure stays positive; its $\mathfrak{b}$-factor solves a genuine Riccati equation.}
\label{fig:cir}
\end{figure}

\section{Control and decision: choosing among futures}\label{sec:control}

Forecasting estimates the future; control changes it. A controlled Markov process has dynamics depending on an admissible action $a_t\in U$:
\begin{equation}\label{eq:controlsde}
\dd X_t=b(X_t,a_t)\dd t+\sigma(X_t,a_t)\dd W_t.
\end{equation}
Given a running cost $c$ and terminal cost $g$, the value function is
\begin{equation}\label{eq:value}
V(t,x)=\inf_{a\in\mathcal{A}}\E_{t,x}\!\Big[\int_t^T c(X_s,a_s)\dd s+g(X_T)\Big].
\end{equation}

\bridge{The dynamic programming principle is the controlled tower property. For small $h$,
\[
V(t,x)=\inf_{a}\E_{t,x}\!\Big[\int_t^{t+h}c(X_s,a_s)\dd s+V(t+h,X_{t+h})\Big],
\]
which is \eqref{eq:value} with the inner conditional expectation collapsed by the tower property of Proposition~\ref{prop:tower}; expanding $V(t+h,X_{t+h})$ by It\^o and letting $h\downarrow0$ gives the Hamilton--Jacobi--Bellman equation.}

\begin{theorem}[HJB equation]\label{thm:hjb}
Under regularity, $V$ solves
\begin{equation}\label{eq:hjb}
-\partial_t V(t,x)=\inf_{a\in U}\big\{c(x,a)+\mathcal{A}^a V(t,x)\big\},\qquad V(T,x)=g(x),
\end{equation}
where $\mathcal{A}^a$ is the generator under action $a$ \citep{FlemingSoner2006,YongZhou1999}.
\end{theorem}

Thus the mathematics of the future has two modes: prediction propagates laws, while control optimizes over possible law propagations.

\begin{outhread}[controlled OU is the linear--quadratic regulator, dual to the filter]\label{ou:control}
Take controlled OU $\dd X_t=(-\theta X_t+a_t)\dd t+\sigma\dd W_t$ with quadratic cost $\int_0^T(qX_s^2+\rho a_s^2)\dd s+g\,X_T^2$. The HJB equation \eqref{eq:hjb} is solved by a quadratic value $V(t,x)=\Pi_t x^2+\text{const}$, optimal linear feedback $a_t^\star=-\rho^{-1}\Pi_t X_t$, and a control Riccati equation
\[
-\dot{\Pi}_t=-2\theta\,\Pi_t+q-\rho^{-1}\Pi_t^2,\qquad \Pi_T=g.
\]
This is the same Riccati structure as the Kalman covariance \eqref{eq:kbric}, transposed and run backward in time: estimation and control are dual linear--quadratic--Gaussian problems, and the OU thread exhibits the two Riccati equations as dual under the standard linear--quadratic--Gaussian transposition/time-reversal correspondence.
\end{outhread}

\section{From model to data: estimation and model risk}\label{sec:estimation}

Every object above has been assumed \emph{known}. In practice the relevant information set, the generator coefficients $(b,\sigma)$, the transition kernel $p_t$, the invariant law $\rho_\infty$, and the change of measure linking $\PP$ to $\QQ$ are all estimated from finite data, and a forecast inherits the error of whatever was estimated.

Three estimation channels recur. (i) \emph{Drift and diffusion estimation}: from discretely sampled paths one estimates $(b,\sigma)$ parametrically or nonparametrically, with a discretization bias that vanishes only as the sampling interval shrinks; the estimated generator $\widehat{\mathcal{A}}$ then drives every downstream forecast. (ii) \emph{Spectral / operator estimation}: because the long-horizon behavior of $P_t$ is controlled by the leading spectrum of $\mathcal{A}$ (Threads~\ref{ou:semigroup},~\ref{ou:gen}), one can estimate eigenvalues and eigenfunctions of the transition operator directly from data, the operator approach to long-term risk of \citet{HansenScheinkman2009}; this is also where data-driven spectral methods meet the semigroup picture. (iii) \emph{Two-measure calibration}: in finance $\QQ$ is calibrated to prices while $\PP$ is estimated from history, and the wedge between them, the market price of risk, is itself an estimated and possibly misspecified quantity.

\begin{proposition}[Generator-error propagation bound]\label{prop:generatorrisk}
Let $P_t=\e^{t\mathcal{A}}$ and $\widehat P_t=\e^{t\widehat{\mathcal{A}}}$ be contraction semigroups on a Banach space $\mathcal{B}$. Suppose, in the bounded-perturbation sense relevant for the chosen core, that $\Delta:=\widehat{\mathcal{A}}-\mathcal{A}$ is bounded with $\norm{\Delta}\le\varepsilon$. Then, for every observable $f\in\mathcal{B}$ and horizon $t\ge0$,
\begin{equation}\label{eq:genrisk}
\norm{\widehat P_t f-P_t f}\le t\varepsilon\norm{f}.
\end{equation}
More generally, if $\norm{\Delta P_s f}\le\varepsilon_f(s)$ for $0\le s\le t$, then
\begin{equation}\label{eq:duhamelbound}
\norm{\widehat P_t f-P_t f}\le\int_0^t \varepsilon_f(s)\dd s.
\end{equation}
\end{proposition}

\begin{proof}
Duhamel's formula gives
\[
\widehat P_t f-P_t f=\int_0^t \widehat P_{t-s}(\widehat{\mathcal{A}}-\mathcal{A})P_s f\dd s.
\]
Taking norms and using contractivity yields \eqref{eq:duhamelbound}; the bounded perturbation case gives \eqref{eq:genrisk}. Thus generator error is controlled no faster than linearly in horizon under a bounded perturbation, before any spectral or nonlinear amplification is considered.
\end{proof}

\begin{remark}[The identities are exact; the inputs are not]
The bridges of this paper hold for the \emph{true} generator. A forecast built on $\widehat{\mathcal{A}}$ carries both variance, from finite samples, and bias, from misspecified dynamics, an unobserved state, or the wrong filtration. The far-future regime is the most fragile: ergodic averages and large-deviation rates depend on the spectral gap and the tails of $\rho_\infty$, so small perturbations of the estimated generator can produce large long-horizon errors (Figure~\ref{fig:risk}). Quantifying this gap between the assumed and the realized law-propagation mechanism is the proper subject of model risk, and it sits outside the classical identities rather than being implied by them.
\end{remark}

\begin{figure}[htbp]
\centering
\includegraphics[width=\textwidth]{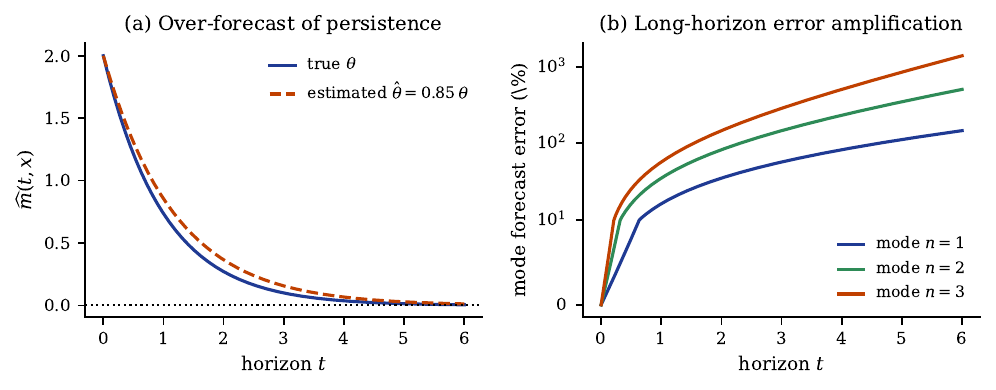}
\caption{Model risk in the OU thread: a $15\%$ \emph{under}-estimate of the mean-reversion rate ($\widehat{\theta}=0.85\,\theta$), with everything else exact. (a) The conditional-mean forecast reverts too slowly, over-stating persistence. (b) The error of the $n$th spectral-mode forecast, $\e^{-n\widehat{\theta}t}/\e^{-n\theta t}-1$, grows with the horizon (symmetric-log scale) and is worse for higher modes: the identities are exact for the true generator, but a forecast inherits the estimate's error, amplified at long horizons.}
\label{fig:risk}
\end{figure}

\subsection{A worked estimator: exact likelihood for the OU thread}\label{sec:ouestimation}
The OU thread closes the loop concretely. Sampled at a fixed step $\Delta$, the OU process \eqref{eq:ousde} is \emph{exactly} a Gaussian first-order autoregression,
\begin{equation}\label{eq:ouar1}
X_{k\Delta}=\mu+\big(X_{(k-1)\Delta}-\mu\big)\rho+\varepsilon_k,\qquad \rho=\e^{-\theta\Delta},\quad \varepsilon_k\sim\mathcal{N}\!\big(0,\,v(\Delta)\big),
\end{equation}
with $v(\Delta)=\sigma^2(1-\rho^2)/2\theta$ and independent innovations, because the transition law of Thread~\ref{ou:semigroup} is Gaussian for every $\Delta$. The exact log-likelihood of a path $X_0,X_\Delta,\dots,X_{N\Delta}$ is therefore the AR(1) likelihood, and maximizing it gives closed-form estimators: $\widehat\rho$ and $\widehat\mu$ from the least-squares regression of $X_{k\Delta}$ on $X_{(k-1)\Delta}$, then
\begin{equation}\label{eq:oumle}
\widehat\theta=-\frac{\log\widehat\rho}{\Delta},\qquad \widehat{\sigma}^2=\frac{2\widehat\theta\,\widehat{v}}{1-\widehat\rho^{\,2}},
\end{equation}
with $\widehat v$ the residual variance. Two lessons generalize. First, OU is the rare diffusion whose discrete-time likelihood is \emph{exact}, so its estimator carries no discretization bias; for a general diffusion the transition density is not available in closed form and one either accepts an Euler pseudo-likelihood with $O(\Delta)$ bias or uses a closed-form likelihood expansion \citep{AitSahalia2002}. The map $\rho\mapsto\theta$ exposes the bias directly: the consistent estimator uses $-\log\widehat\rho/\Delta$, whereas the naive Euler estimator effectively uses $(1-\widehat\rho)/\Delta$, and the two differ by the deterministic factor $(1-\rho)/(-\log\rho)$, which departs from $1$ as $\theta\Delta$ grows. Second, the estimator's precision is set by the Fisher information of \eqref{eq:ouar1}: by the delta method $\Var(\widehat\theta)\approx(\e^{2\theta\Delta}-1)/(N\Delta^2)$, which is $O(1/N)$ at fixed $\Delta$ and tends to the continuous-record variance $2\theta/T$ as $\Delta\downarrow0$ with $N\Delta=T$. The generator coefficients are thus estimable at the parametric rate, but the long-horizon objects built from them remain subject to the amplification of Proposition~\ref{prop:generatorrisk} and Figure~\ref{fig:risk}. Figure~\ref{fig:estimation} shows both effects.

\begin{figure}[htbp]
\centering
\includegraphics[width=\textwidth]{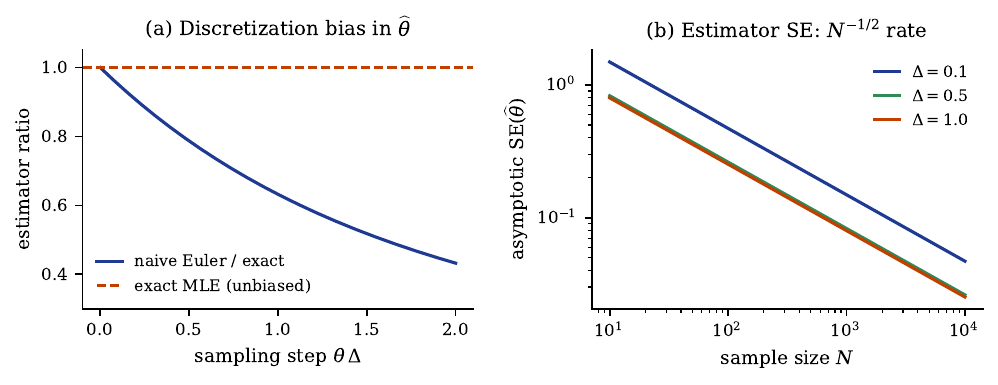}
\caption{Estimation for the OU thread, computed from closed forms (no simulation). (a) Discretization bias of the naive Euler estimator of $\theta$: the factor $(1-\e^{-\theta\Delta})/(\theta\Delta)$ relative to the exact estimator departs from $1$ as the sampling step $\theta\Delta$ grows, while the exact likelihood \eqref{eq:oumle} is free of discretization bias at every $\Delta$. (b) Asymptotic standard error of the exact estimator $\widehat\theta$, $\sqrt{(\e^{2\theta\Delta}-1)/(N\Delta^2)}$, against sample size $N$ for several steps $\Delta$ (log--log), showing the parametric $N^{-1/2}$ rate.}
\label{fig:estimation}
\end{figure}

\section{The far future: ergodicity and large deviations}\label{sec:farfuture}

Two asymptotic regimes describe long horizons: the typical future and the improbable future.

\subsection{The typical long run}
A measure $\mu$ is invariant for $(P_t)$ if $\mu P_t=\mu$ for all $t$, equivalently $\int P_tf\dd\mu=\int f\dd\mu$ for bounded measurable $f$. For a diffusion with invariant density this is the stationary Fokker--Planck equation $\mathcal{A}^*\mu=0$.

\begin{theorem}[Ergodic theorem for Markov processes]\label{thm:ergodic}
Let $X$ be a positive Harris recurrent Markov process with invariant probability measure $\mu$. If $f\in L^1(\mu)$, then for every initial distribution absolutely continuous with respect to the maximal irreducibility measure, and in particular for $\mu$-almost every initial state,
\begin{equation}\label{eq:ergodic}
\frac{1}{T}\int_0^T f(X_s)\dd s\ \xrightarrow[T\to\infty]{\text{a.s.}}\ \int_E f\dd\mu.
\end{equation}
Under stronger stability assumptions, such as positive Harris recurrence together with the usual irreducibility and aperiodicity conditions on the relevant class, analogous convergence statements hold beyond the stationary initial distribution.
\end{theorem}

Equation \eqref{eq:ergodic} states that the infinite future forgets the initial condition: time averages converge to spatial averages against the invariant measure, the only forecast that survives at the longest horizons \citep{MeynTweedie2009,DaPratoZabczyk1996}.

\subsection{The improbable future}
Risk analysis also asks how unlikely a given excursion is.

\begin{definition}[Large-deviation principle]
A family $(\mu_\varepsilon)_{\varepsilon>0}$ on a Polish space satisfies a large-deviation principle with speed $\varepsilon^{-1}$ and good rate function $I$ if, for every closed $F$ and open $G$,
\[
\limsup_{\varepsilon\downarrow0}\varepsilon\log\mu_\varepsilon(F)\le-\inf_{x\in F}I(x),\qquad
\liminf_{\varepsilon\downarrow0}\varepsilon\log\mu_\varepsilon(G)\ge-\inf_{x\in G}I(x).
\]
Heuristically $\mu_\varepsilon(\Gamma)\approx\exp(-\varepsilon^{-1}\inf_{x\in\Gamma}I(x))$ for regular $\Gamma$.
\end{definition}

For small-noise diffusions $\dd X^\varepsilon_t=b(X^\varepsilon_t)\dd t+\sqrt{\varepsilon}\,\sigma(X^\varepsilon_t)\dd W_t$, Freidlin--Wentzell theory identifies $I$ as an action functional. Rare futures concentrate around minimizers of this action, so the rate function is a variational description of tail risk: it gives not only the probability of an unlikely event but the most likely path by which it occurs \citep{DemboZeitouni1998,FreidlinWentzell2012}.

\begin{outhread}[invariant law, ergodic average, and the action functional]\label{ou:far}
For OU the invariant measure is the Gaussian \eqref{eq:ouinvariant}, $\mathcal{N}(\mu,\sigma^2/2\theta)$; ergodic averages of any $f\in L^1$ converge to its Gaussian mean, and the $\e^{-\theta t}$ relaxation of Thread~\ref{ou:wasserstein} sets the rate. For small-noise OU $\dd X^\varepsilon_t=\theta(\mu-X^\varepsilon_t)\dd t+\sqrt{\varepsilon}\,\sigma\dd W_t$, the Freidlin--Wentzell action of an absolutely continuous path $\varphi$ on $[0,T]$ with $\varphi_0=x_0$ is
\begin{equation}\label{eq:fwaction}
I_{0,T}(\varphi)=\frac{1}{2\sigma^2}\int_0^T\big|\dot{\varphi}_t-\theta(\mu-\varphi_t)\big|^2\dd t,
\end{equation}
finite exactly for paths absolutely continuous from $x_0$; its minimizers are the most likely excursions, and the induced quasipotential governs the exponential rarity of large deviations from $\mu$.
\end{outhread}

\begin{cirthread}[the non-Gaussian far future]\label{cir:far}
For CIR the invariant law is the Gamma measure \eqref{eq:cirinvariant}; ergodic averages of $f\in L^1(p_\infty)$ converge to its Gamma mean, and the spectral gap $\theta$ again sets the rate (Thread~\ref{cir:kolmogorov}). The small-noise large-deviation rate function takes the same Freidlin--Wentzell form as \eqref{eq:fwaction} with the state-dependent diffusion $\sigma^2\varphi$ in place of $\sigma^2$, so the metric weighting the action degenerates toward the boundary --- the long-run theory inherits the positivity constraint that distinguishes CIR from OU.
\end{cirthread}

\section{Discrete time: the same spine on a lattice}\label{sec:discrete}

Most forecasting and machine-learning practice is in discrete time. The entire spine transfers, with the generator replaced by the one-step transition operator minus the identity and It\^o's formula replaced by the Doob decomposition. Let $(\mathcal{F}_n)_{n\ge0}$ be a filtration and $X$ a Markov chain with transition operator $(Pf)(x)=\E[f(X_{n+1})\mid X_n=x]$.

\begin{proposition}[Discrete Doob decomposition]\label{prop:doob}
Any integrable adapted $X$ admits the unique decomposition $X_n=X_0+A_n+M_n$ with $A_n=\sum_{k=1}^n\E[X_k-X_{k-1}\mid\mathcal{F}_{k-1}]$ predictable, $A_0=0$, and $M_n=\sum_{k=1}^n\big(X_k-X_{k-1}-\E[X_k-X_{k-1}\mid\mathcal{F}_{k-1}]\big)$ a martingale. The increments $\xi_k:=M_k-M_{k-1}$ are martingale differences, $\E[\xi_k\mid\mathcal{F}_{k-1}]=0$.
\end{proposition}

This is the exact discrete analogue of It\^o's split \eqref{eq:ito}: predictable drift plus martingale innovation. Table~\ref{tab:dictionary} gives the full dictionary; each row is the same conditional-law statement read in two time scales \citep{Norris1997,CappeMoulinesRyden2005}.

\begin{table}[tbp]
\centering
\renewcommand{\arraystretch}{1.32}
\begin{tabular}{@{}L{0.30\textwidth} L{0.32\textwidth} L{0.30\textwidth}@{}}
\toprule
\textbf{Object} & \textbf{Continuous time} & \textbf{Discrete time}\\
\midrule
Information & filtration $(\mathcal{F}_t)$ & filtration $(\mathcal{F}_n)$\\
Optimal forecast & $\E[\cdot\mid\mathcal{F}_t]$ ($L^2$-projection) & $\E[\cdot\mid\mathcal{F}_n]$\\
Propagation (observables) & semigroup $P_t$ & powers $P^n$\\
Propagation (laws) & $\mu\mapsto\mu P_t$ & $\mu\mapsto\mu P^n$\\
Local dynamics & generator $\mathcal{A}=\lim_{t\downarrow0}(P_t-\Id)/t$ & discrete generator $P-\Id$\\
Backward equation & $\partial_t u=\mathcal{A}u$ & $u_{n+1}-u_n=(P-\Id)u_n$\\
Forward equation & $\partial_t p=\mathcal{A}^*p$ & $\mu_{n+1}-\mu_n=\mu_n(P-\Id)$\\
Pathwise model & SDE \eqref{eq:sde}, It\^o & recursion $X_{n+1}=F(X_n,\xi_{n+1})$\\
Trend $+$ innovation & It\^o split \eqref{eq:ito} & Doob split (Prop.~\ref{prop:doob})\\
Innovation noise & It\^o integral $\int H\dd W$ & martingale differences $\xi_k$\\
Filtering & Kushner--Stratonovich \eqref{eq:ks}; Kalman--Bucy & forward (HMM) recursion; Kalman filter\\
Pricing & $\theta_t V_t$ a $\QQ$-martingale & $V_n=\E^{\QQ}[V_{n+1}/(1+r_n)\mid\mathcal{F}_n]$\\
Control & HJB \eqref{eq:hjb} & Bellman $V_n(x)=\inf_a\{c+\E[V_{n+1}\mid x,a]\}$\\
Typical far future & ergodic theorem \eqref{eq:ergodic} & chain ergodic theorem\\
Improbable far future & Freidlin--Wentzell \eqref{eq:fwaction} & Cram\'er / G\"artner--Ellis\\
\bottomrule
\end{tabular}
\caption{The same spine in continuous and discrete time. Each row is one conditional-law statement; the generator becomes $P-\Id$ and It\^o becomes the Doob decomposition.}
\label{tab:dictionary}
\end{table}

\section{A unified diagram}\label{sec:diagram}

Figure~\ref{fig:spine} draws the spine and its branches. The vertical axis is the core progression of \S\ref{sec:info}--\S\ref{sec:kolmogorov}; the branches are the specializations (innovation, geometry, partial observation, valuation, control, far future) attached at the link where each enters.

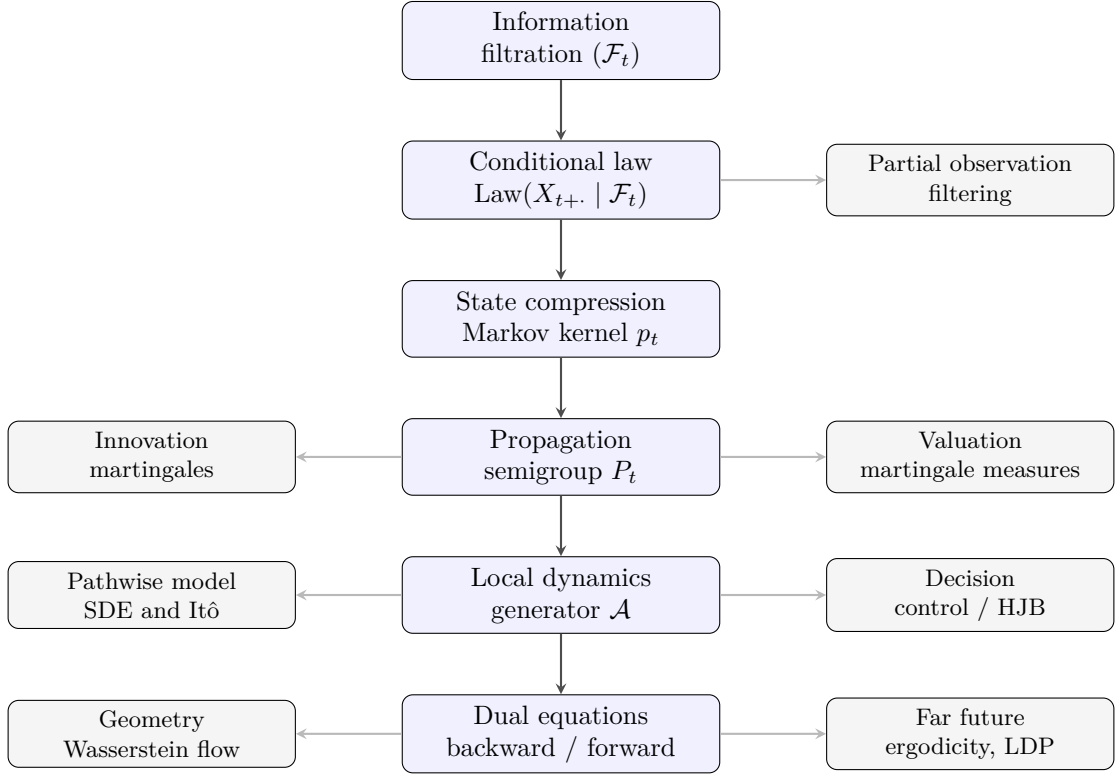
\begin{figure}[htbp]
\centering
\begin{tikzpicture}[
  >=stealth,
  node distance=8mm and 14mm,
  core/.style={draw,rounded corners,align=center,fill=blue!6,inner sep=4pt,minimum width=42mm,font=\small},
  branch/.style={draw,rounded corners,align=center,fill=gray!8,inner sep=4pt,minimum width=38mm,font=\footnotesize},
  ar/.style={->,thick,gray!60!black},
  bar/.style={->,thick,gray!55}
]
\node[core] (info)  {Information\\ filtration $(\mathcal{F}_t)$};
\node[core,below=of info] (law) {Conditional law\\ $\Law(X_{t+\cdot}\mid\mathcal{F}_t)$};
\node[core,below=of law] (kernel) {State compression\\ Markov kernel $p_t$};
\node[core,below=of kernel] (semi) {Propagation\\ semigroup $P_t$};
\node[core,below=of semi] (gen) {Local dynamics\\ generator $\mathcal{A}$};
\node[core,below=of gen] (pde) {Dual equations\\ backward / forward};

\draw[ar] (info)--(law);
\draw[ar] (law)--(kernel);
\draw[ar] (kernel)--(semi);
\draw[ar] (semi)--(gen);
\draw[ar] (gen)--(pde);

\node[branch,right=of law] (filter) {Partial observation\\ filtering};
\node[branch,right=of semi] (price) {Valuation\\ martingale measures};
\node[branch,right=of gen] (control) {Decision\\ control / HJB};
\node[branch,right=of pde] (far) {Far future\\ ergodicity, LDP};

\draw[bar] (law)--(filter);
\draw[bar] (semi)--(price);
\draw[bar] (gen)--(control);
\draw[bar] (pde)--(far);

\node[branch,left=of gen] (sde) {Pathwise model\\ SDE and It\^o};
\node[branch,left=of semi] (mart) {Innovation\\ martingales};
\node[branch,left=of pde] (geo) {Geometry\\ Wasserstein flow};

\draw[bar] (gen)--(sde);
\draw[bar] (semi)--(mart);
\draw[bar] (pde)--(geo);
\end{tikzpicture}
\caption{The unified spine (center) and its branches. Conditional expectation projects the conditional law into a point forecast; the Markov property compresses the past into the present; the semigroup and generator propagate it; backward and forward equations are adjoint; and the martingale component is the part of the future no admissible information removes.}
\label{fig:spine}
\end{figure}

\section{Synthesis}\label{sec:synthesis}

The pieces assemble into a single mathematical picture. The following theorem records the common core of the paper in one statement.

\begin{theorem}[The conditional-law spine]\label{thm:conditional-law-spine}
Let $X=(X_t)_{t\ge0}$ be a time-homogeneous Feller--Markov process on a Polish state space $E$, with transition semigroup $(P_t)_{t\ge0}$ and generator $(\mathcal{A},\mathcal{D}(\mathcal{A}))$. For every bounded measurable observable $f$ and every $s,t\ge0$,
\[
\E\!\left[f(X_{t+s})\mid\mathcal{F}_t\right]=(P_s f)(X_t),
\qquad
P_{t+s}=P_tP_s.
\]
If $f\in\mathcal{D}(\mathcal{A})$, then $u(s,x)=P_s f(x)$ satisfies the backward Kolmogorov equation in the semigroup sense,
\[
\partial_s u(s,x)=\mathcal{A}u(s,x),
\qquad
u(0,x)=f(x).
\]
Dually, for every initial law $\mu_0$, the propagated laws $\mu_s=\mu_0P_s$ satisfy the weak forward equation
\[
\frac{\dd}{\dd s}\int_E f(x)\,\mu_s(\dd x)
=
\int_E \mathcal{A}f(x)\,\mu_s(\dd x),
\qquad f\in\mathcal{D}(\mathcal{A}).
\]
When $\mu_s$ admits a density $p_s$ and the formal adjoint $\mathcal{A}^*$ is well defined on that density class, this weak identity is the forward Kolmogorov equation
\[
\partial_s p_s=\mathcal{A}^*p_s.
\]
Thus the conditional law of the future, the Markov semigroup, the generator, and the backward/forward Kolmogorov equations are not separate modeling languages; they are equivalent representations of the same forecasting object.
\end{theorem}

\begin{proof}
The Markov property gives
\[
\E[f(X_{t+s})\mid\mathcal{F}_t]=\E[f(X_{t+s})\mid X_t]=(P_sf)(X_t),
\]
and applying the same identity first over horizon $s$ and then over horizon $t$ gives Chapman--Kolmogorov, equivalently $P_{t+s}=P_tP_s$. By definition of the generator,
\[
\mathcal{A}f=\lim_{h\downarrow0}\frac{P_hf-f}{h},
\]
so for $u(s,\cdot)=P_s f$ the semigroup law gives
\[
\partial_su(s,\cdot)=\lim_{h\downarrow0}\frac{P_{s+h}f-P_sf}{h}
=P_s\mathcal{A}f=\mathcal{A}P_sf,
\]
with the equality interpreted on the appropriate generator domain. Finally,
\[
\int_E f\,\dd\mu_s=\int_E P_sf\,\dd\mu_0,
\]
and differentiating in $s$ yields the weak forward equation. If densities exist, integration by parts identifies the weak equation with $\partial_sp_s=\mathcal{A}^*p_s$.
\end{proof}

\begin{enumerate}\itemsep2pt
\item Information is a filtration $(\mathcal{F}_t)$. The future is the conditional law of the path beyond $t$ given $\mathcal{F}_t$.
\item The optimal point forecast is conditional expectation, the $L^2$-projection onto present information (Theorem~\ref{thm:projection}). The tower property keeps forecasts dynamically consistent and, specialized to the Markov filtration, is the semigroup law itself.
\item The full distributional forecast is a regular conditional law; it contains all conditional expectations of bounded observables.
\item The Markov assumption compresses the relevant past into the present state, producing kernels satisfying Chapman--Kolmogorov.
\item The semigroup propagates observables by $P_tf(x)=\E_x[f(X_t)]$ and propagates laws dually by $\mu\mapsto\mu P_t$.
\item The generator $\mathcal{A}=\lim_{t\downarrow0}(P_t-\Id)/t$ is the infinitesimal law of motion; the semigroup is its integrated flow.
\item Kolmogorov equations convert the generator into PDEs, backward for observables and forward for densities, and the two are exact adjoints by integration by parts.
\item Geometrically, the forward equation is the Wasserstein gradient flow of free energy (Theorem~\ref{thm:jko}); for OU, mean reversion is $W_2$ contraction at the spectral-gap rate $\theta$. Forecasting, estimation, and law-evolution are the Hilbert, Fisher--Rao, and Wasserstein faces of one geometry (\S\ref{sec:threegeom}).
\item Pathwise dynamics represent the same law by an SDE; It\^o's formula splits every smooth observable into predictable drift plus martingale innovation, and taking expectations returns the backward equation.
\item Martingales are the residue of unpredictability. Doob--Meyer separates trend from surprise; martingale representation expresses surprise through driving noise.
\item Filtering replaces direct observation by conditional laws given noisy data; innovations are the observation-level martingale surprises, and the Kalman and control Riccati equations are dual under the standard linear--quadratic--Gaussian transposition/time-reversal correspondence.
\item Finance prices discounted payoffs as martingales under an equivalent measure; in Markov models this is Feynman--Kac, hence a backward PDE, hence the explicit Vasicek yield curve for OU.
\item Control moves from forecasting to choosing among futures through the dynamic programming principle, the controlled tower property, and the HJB equation.
\item Asymptotics describe the far future: ergodicity gives typical long-run averages, large deviations quantify improbable excursions, and both are governed by the spectrum of $\mathcal{A}$.
\item The same spine carries two contrasting closed-form threads: OU (Gaussian, real line, Hermite spectrum, linear term structure) and CIR (non-Gaussian, half-line, Laguerre spectrum, Riccati term structure), with identical drift, mean reversion, and relaxation rates $-n\theta$ (Table~\ref{tab:oucir}) --- evidence that the construction is not Gaussian-specific.
\item Every object is in practice estimated; the identities are exact for the true generator (Proposition~\ref{prop:generatorrisk} bounds the error of an estimated one), and the OU thread is estimable in closed form at the parametric rate while its long-horizon objects remain fragile (\S\ref{sec:estimation}).
\end{enumerate}

In one line: the present information set determines a conditional law; under Markov compression a generator determines the forward evolution of that law; conditional expectation projects it into a forecast; the martingale component is the part of the future no admissible information can remove; and in practice the generator itself is estimated, so the forecast is only as good as that estimate.

\paragraph{On rigor.} The statements above are standard theorems of probability, stochastic analysis, semigroup theory, optimal transport, filtering, mathematical finance, control, ergodic theory, and large deviations. For readability several results are stated in representative forms rather than maximal generality. In applications the technical hypotheses matter: completeness and right-continuity of filtrations; existence of regular conditional probabilities; domain and boundary conditions for generators; Lipschitz or weak-solution assumptions for SDEs; integrability and class (D) conditions for Doob--Meyer; smoothness and convexity of $\Psi$ for the Wasserstein gradient-flow and contraction statements; local boundedness or sigma-martingale refinements for the fundamental theorem of asset pricing; observability and stabilizability for filtering and control; recurrence and irreducibility for ergodicity; and exponential tightness for large deviations. The central message is invariant across these variants: forecasting is conditional law, dynamics is law propagation, surprise is martingale residual, and the link to data is an estimated generator.

\appendix
\section{Selected derivations}\label{app:derivations}

The bridges of the main text are recorded here in full for the diffusion case, and the closed forms asserted in the Ornstein--Uhlenbeck thread are derived. Throughout, boundary terms vanish under the standing assumption that test functions and densities decay at infinity.

\subsection{The forward operator as the adjoint of the generator}\label{app:adjoint}
Let $\mathcal{A}f=bf'+\tfrac12 a f''$ in one dimension, with $a=\sigma^2$. For $f,g$ smooth and decaying, two integrations by parts give
\[
\int b f' g\dd x=-\int f\,(bg)'\dd x,\qquad
\int \tfrac12 a f'' g\dd x=\int f\,\tfrac12 (ag)''\dd x,
\]
so that $\inner{\mathcal{A}f}{g}=\int f\big(-(bg)'+\tfrac12(ag)''\big)\dd x=\inner{f}{\mathcal{A}^*g}$ with
\[
\mathcal{A}^*g=-(bg)'+\tfrac12(ag)'',
\]
which is the one-dimensional forward operator of \eqref{eq:forward}. The forward Kolmogorov equation $\partial_t p=\mathcal{A}^*p$ is therefore the adjoint of the backward equation $\partial_t u=\mathcal{A}u$, exactly as the duality \eqref{eq:duality} requires.

\subsection{The Ornstein--Uhlenbeck transition law and invariant measure}\label{app:ou}
Write $Y_t=X_t-\mu$, so $\dd Y_t=-\theta Y_t\dd t+\sigma\dd W_t$. The integrating factor $\e^{\theta t}$ gives $\dd(\e^{\theta t}Y_t)=\sigma\e^{\theta t}\dd W_t$, hence
\[
X_t=\mu+(X_0-\mu)\e^{-\theta t}+\sigma\!\int_0^t\e^{-\theta(t-s)}\dd W_s,
\]
which is \eqref{eq:ousol}. The stochastic integral has a deterministic integrand, so $X_t$ is Gaussian with mean $\mu+(x-\mu)\e^{-\theta t}$ and, by the It\^o isometry,
\[
\Var(X_t\mid X_0=x)=\sigma^2\!\int_0^t\e^{-2\theta(t-s)}\dd s=\frac{\sigma^2}{2\theta}\big(1-\e^{-2\theta t}\big),
\]
recovering $m(t,x)$ and $v(t)$ of Thread~\ref{ou:semigroup}. For the invariant law, set the stationary probability flux $J=\theta(\mu-y)p-\tfrac{\sigma^2}{2}p'$ to zero. The choice $p_\infty(y)\propto\exp(-\theta(y-\mu)^2/\sigma^2)$ gives $p_\infty'/p_\infty=-2\theta(y-\mu)/\sigma^2$, whence $J=\theta(\mu-y)p_\infty+\theta(y-\mu)p_\infty=0$ and so $\mathcal{A}^*p_\infty=-J'=0$. Thus $p_\infty=\mathcal{N}(\mu,\sigma^2/2\theta)$ as in \eqref{eq:ouinvariant}.

\subsection{The Vasicek term structure from the affine ansatz}\label{app:vasicek}
With $\tau=T-t$ and the ansatz $P=\exp(\mathfrak{a}(\tau)-\mathfrak{b}(\tau)r)$, $\mathfrak{a}(0)=\mathfrak{b}(0)=0$, the derivatives are $\partial_t P=-(\mathfrak{a}'-\mathfrak{b}'r)P$, $\partial_r P=-\mathfrak{b}P$, $\partial_{rr}P=\mathfrak{b}^2P$. Substituting into the term-structure PDE of Thread~\ref{ou:pricing} and dividing by $P$,
\[
-(\mathfrak{a}'-\mathfrak{b}'r)-\theta(\mu-r)\mathfrak{b}+\tfrac{\sigma^2}{2}\mathfrak{b}^2-r=0.
\]
Matching powers of $r$ separates this into
\[
\mathfrak{b}'=1-\theta\mathfrak{b},\qquad \mathfrak{a}'=\tfrac{\sigma^2}{2}\mathfrak{b}^2-\theta\mu\,\mathfrak{b}.
\]
The first is linear with $\mathfrak{b}(0)=0$, giving $\mathfrak{b}(\tau)=(1-\e^{-\theta\tau})/\theta$; substituting and integrating the second yields the closed form
\[
\mathfrak{a}(\tau)=\Big(\mu-\frac{\sigma^2}{2\theta^2}\Big)\big(\mathfrak{b}(\tau)-\tau\big)-\frac{\sigma^2}{4\theta}\,\mathfrak{b}(\tau)^2,
\]
which is the Vasicek bond price used in Thread~\ref{ou:pricing} (direct differentiation confirms $\mathfrak{a}'=\tfrac{\sigma^2}{2}\mathfrak{b}^2-\theta\mu\mathfrak{b}$ using $\mathfrak{b}'-1=-\theta\mathfrak{b}$).

\subsection{The linear--quadratic regulator and its Riccati equation}\label{app:lqr}
For controlled OU $\dd X_t=(-\theta X_t+a_t)\dd t+\sigma\dd W_t$ and cost density $qx^2+\rho a^2$, the HJB equation \eqref{eq:hjb} reads
\[
-\partial_t V=\inf_{a}\Big\{qx^2+\rho a^2+(-\theta x+a)\partial_x V+\tfrac{\sigma^2}{2}\partial_{xx}V\Big\}.
\]
The inner minimization gives $a^\star=-\partial_x V/(2\rho)$; substituting,
\[
-\partial_t V=qx^2-\frac{(\partial_x V)^2}{4\rho}-\theta x\,\partial_x V+\tfrac{\sigma^2}{2}\partial_{xx}V.
\]
The quadratic ansatz $V(t,x)=\Pi_t x^2+\varphi_t$ has $\partial_x V=2\Pi_t x$ and $\partial_{xx}V=2\Pi_t$. Matching the coefficient of $x^2$ gives the control Riccati equation
\[
-\dot{\Pi}_t=-2\theta\Pi_t+q-\rho^{-1}\Pi_t^2,\qquad \Pi_T=g,
\]
and the optimal feedback $a^\star_t=-\rho^{-1}\Pi_t X_t$, as in Thread~\ref{ou:control}; the constant term obeys $\dot{\varphi}_t=-\sigma^2\Pi_t$, $\varphi_T=0$. Comparing with the Kalman covariance equation \eqref{eq:kbric} exhibits the two Riccati equations as dual under time reversal and transposition, the linear--quadratic--Gaussian duality of estimation and control.

\subsection{The CIR invariant measure and term-structure Riccati}\label{app:cir}
For CIR the forward operator is $\mathcal{A}^*p=-\partial_y(\theta(\mu-y)p)+\tfrac12\sigma^2\partial_{yy}(yp)$. Setting the stationary flux $J=\theta(\mu-y)p-\tfrac12\sigma^2\partial_y(yp)$ to zero and trying $p_\infty(y)=Cy^{\alpha-1}\e^{-\beta y}$ gives $\partial_y(yp_\infty)=p_\infty(\alpha-\beta y)$, so
\[
J=\theta(\mu-y)p_\infty-\tfrac12\sigma^2 p_\infty(\alpha-\beta y)
=p_\infty\Big[\big(\theta\mu-\tfrac12\sigma^2\alpha\big)-\big(\theta-\tfrac12\sigma^2\beta\big)y\Big].
\]
Both brackets vanish for $\beta=2\theta/\sigma^2$ and $\alpha=2\theta\mu/\sigma^2$, so $J\equiv0$ and $\mathcal{A}^*p_\infty=-\partial_y J=0$: the invariant law is $\mathrm{Gamma}(\alpha,\beta)$, equation \eqref{eq:cirinvariant}, with mean $\alpha/\beta=\mu$ and variance $\alpha/\beta^2=\mu\sigma^2/2\theta$. For the bond price, the affine ansatz $P=\exp(\mathfrak{a}(\tau)-\mathfrak{b}(\tau)r)$ in the CIR term-structure PDE of Thread~\ref{cir:pricing} gives, on collecting powers of $r$ (the $\tfrac12\sigma^2 r\,\partial_{rr}P$ term now contributes $\tfrac12\sigma^2 r\,\mathfrak{b}^2$ to the $r$-coefficient),
\[
\mathfrak{b}'=1-\theta\mathfrak{b}-\tfrac12\sigma^2\mathfrak{b}^2,\qquad \mathfrak{a}'=-\theta\mu\,\mathfrak{b},\qquad \mathfrak{a}(0)=\mathfrak{b}(0)=0.
\]
The $\mathfrak{b}$-equation is a constant-coefficient Riccati; with $\gamma=\sqrt{\theta^2+2\sigma^2}$ its solution is $\mathfrak{b}(\tau)=2(\e^{\gamma\tau}-1)/[(\gamma+\theta)(\e^{\gamma\tau}-1)+2\gamma]$ (one checks $\mathfrak{b}(0)=0$, $\mathfrak{b}'(0)=1$, and the steady state $\tfrac12\sigma^2\mathfrak{b}^2+\theta\mathfrak{b}-1=0$), and integrating $\mathfrak{a}'=-\theta\mu\mathfrak{b}$ gives the closed form quoted in Thread~\ref{cir:pricing} \citep{CIR1985}. The quadratic term, absent for Vasicek, is exactly the contribution of the $\sqrt{r}$ diffusion.

\bibliographystyle{plainnat}
\bibliography{modeling_future_refs}

\end{document}